\documentclass[lettersize,journal]{IEEEtran}
\usepackage{amsmath,amsfonts}
\usepackage{algorithmic}
\usepackage{algorithm}
\usepackage{array}
\usepackage{bm}
\usepackage{textcomp}
\usepackage{stfloats}
\usepackage{url}
\usepackage{verbatim}
\usepackage{slashbox}
\usepackage{graphicx}
\usepackage{cite}
\usepackage{xcolor}
\usepackage[font=small,labelfont=bf]{caption}
\usepackage[font=small,labelfont=bf]{subcaption}
\usepackage{makecell}
\usepackage{multicol}
\usepackage{multirow}
\usepackage{booktabs}
\newtheorem{theorem}{Theorem}

\newenvironment{proof}{ \textbf{Proof:} }{ \hfill $\Box$}

\def\bb0{{\mathbb{0}}}

\def\bb{{\mathbf{b}}}

\def\b0{{\mathbf{0}}}

\def\bI{{\mathbf{I}}}

\def\bbC{{\mathbb{C}}}

\def\bbE{{\mathbb{E}}}

\def\cL{\mathcal{L}}

\def\cR{\mathcal{R}}

\def\cT{\mathcal{T}}

\def\sfA{\mathsf{A}}

\def\sfS{\mathsf{S}}

\def\sfZ{\mathsf{Z}}

\def\sfa{{\mathsf{a}}}
\def\sfb{{\mathsf{b}}}

\def\sfj{{\mathsf{j}}}

\def\sfs{{\mathsf{s}}}

\def\sf0{{\mathsf{0}}}

\def\rm0{{\mathrm{0}}}

\def\b0{{\pmb{0}}} 

\begin{document}
	
	\title{	A generalization of the achievable rate of a MISO system using  Bode-Fano wideband matching theory }

	\author{\IEEEauthorblockN{Nitish  Deshpande, \textit{Student Member, IEEE,} Miguel R. Castellanos, \textit{Member, IEEE,} Saeed R. Khosravirad, \textit{Member, IEEE,} Jinfeng Du,   \textit{Member, IEEE,}   Harish Viswanathan, \textit{Fellow, IEEE,} and Robert W. Heath Jr., \textit{Fellow, IEEE}}\\ 
		\thanks{ Nitish  Deshpande,  Miguel R. Castellanos, and Robert W. Heath Jr. are
			with the Department of Electrical and Computer Engineering, North Carolina State University, Raleigh, NC 27606 USA (email: \{nvdeshpa, mrcastel, rwheathjr\}@ncsu.edu).
			Saeed R. Khosravirad, Jinfeng Du, and Harish Viswanathan are with Nokia
			Bell Laboratories, Murray Hill, NJ 07974, USA (email: \{saeed.khosravirad,  jinfeng.du,
			harish.viswanathan\}@nokia-bell-labs.com).
			This project is funded in part by Nokia Bell Laboratories, Murray Hill, NJ 07974, USA. This material is also based upon work supported in part by the National Science Foundation under Grant No. NSF-CNS-2147955 and NSF-ECCS-2153698 and by funds from federal agencies and industry partners as specified in the Resilient \& Intelligent NextG Systems (RINGS) program.
		}
	}
	
	
	
	
	\maketitle
	
	\begin{abstract}
		Impedance-matching networks affect power transfer from  the radio frequency (RF) chains to the antennas. Their design impacts the signal to noise ratio (SNR) and the achievable rate.
		In this paper, we maximize the information-theoretic achievable rate of a multiple-input-single-output (MISO) system with  wideband matching constraints. Using a multiport circuit theory approach with frequency-selective scattering parameters, we propose a general framework for optimizing the MISO achievable rate that incorporates Bode-Fano wideband matching theory.
		 We express the solution to the achievable rate optimization problem in terms of the optimized transmission coefficient and the Lagrangian parameters corresponding to the Bode-Fano inequality constraints. We apply this framework to a single electric Chu's   antenna and an array of two electric Chu's antennas. We compare the optimized achievable rate obtained numerically with other benchmarks like the ideal achievable rate computed by disregarding matching constraints and the achievable rate obtained by using sub-optimal matching strategies like conjugate matching and frequency-flat transmission. We also propose a practical methodology to approximate the  achievable rate bound by using the  optimal transmission coefficient to derive a physically realizable matching network through the ADS software.
	\end{abstract}
	
	
	\section{Introduction}
	
	As wireless systems exploit higher bandwidths, it is crucial to design matching networks that  support the desired  power transfer in the band of interest to achieve the target data rate\cite{buon_kiong_lau_impact_2006, nie_bandwidth_2017,taluja_bandwidth_2011,saab_optimizing_2022}. For narrowband systems, matching networks are optimized for power transfer between source and load at a  single  frequency. For wideband arrays, it is challenging to design matching networks  because the load depends on the  frequency-selectivity of 
	the  array including mutual coupling between antennas~\cite{saab_optimizing_2022}.
		The Bode-Fano theory captures these practical  matching constraints  with a frequency-selective circuit theory approach based on scattering parameters\cite{nie_bandwidth_2017,steer2019microwave}.  
		In this paper, we incorporate these constraints in achievable rate  analysis unlike conventional wideband MIMO literature which does not treat matching networks as a part of the analysis\cite{park_dynamic_2017, heath_jr_foundations_2018}.

		The problem of matching a source impedance to a load impedance in conventional RF literature is mostly based on power transfer based metrics\cite{steer2019microwave, bode1945network, fano1950theoretical, 8031276, nie_systematic_2014}. In general,
		the transmit matching network is designed to maximize power transfer efficiency  while the receive matching network to minimize  the noise figure.		
	For a narrowband system, the conjugate matching network  is designed such that the effective load impedance  equals  complex conjugate of the source impedance. 
	For broadband matching, the constant quality factor circle technique can be used~\cite{steer2019microwave}. Designing matching networks for systems operating at higher fractional bandwidths is challenging because of the frequency-selectivity of the load.
	Recently, a globally optimal approach to designing wideband matching networks defined a unique trajectory connecting source and load on Smith chart using the power transfer efficiency metric
	\cite{8031276}.  
	Although power transfer efficiency is important, it only quantifies the power transfer from the RF chain to the  antennas within  a transmitter.
	From a communication theoretic perspective, the most relevant metric is the end-to-end achievable data rate. The communication rate depends on  factors like the bandwidth, wireless  propagation channel, beamforming response at both receiver and transmitter, mutual coupling between antenna elements, and their radiation patterns. The achievable rate metric captures all factors.
	Hence, wireless devices should optimize the  matching network to maximize the rate rather than the power transfer efficiency.

	For analysis of wideband systems, it is essential to understand the fundamental design tradeoffs between gain and bandwidth\cite{10002944}. 
	In large phased-arrays operating at higher fractional bandwidths, there is a  phase mismatch between the frequency-flat phase-shifter and frequency-selective array response\cite{7841766}. 
		The frequency-selectivity of antennas and matching networks was considered for analysis of dense array wideband massive MIMO\cite{deshpande2023analysis}.  The results in \cite{deshpande2023analysis,taluja_diversity_2013} showed that  for  systems which use matching networks based on the   conjugate matching strategy, the SNR drops drastically for frequencies away from the center frequency. 
			A  matching network based on a narrowband assumption is sub-optimal in a wideband setting. Hence, it is necessary to optimize physically realizable matching networks over the bandwidth of interest.
		These examples show that the shift from frequency-flat to frequency-selective models is necessary as wireless systems transition from narrowband to wideband operation\cite{8354789, 7841766, 10002944, deshpande2023analysis,taluja_diversity_2013  }.
		
		The circuit theory approach to modeling wireless communication systems enables  incorporating physically consistent frequency-selective models for the antennas, arrays, wireless channel, and the  RF chain components in the analysis\cite{ivrlac_toward_2010}.
	This approach  captures effects like mutual coupling in the form of impedance or scattering matrices thus making the system analysis more realistic and tractable\cite{saab_optimizing_2022, 9097141, 9048753,ivrlac_toward_2010,taluja_diversity_2013, wallace_mutual_2004, kundu_enhancing_2017,  shyianov_achievable_2022, akrout2022super, akrout2023physically}. 
	Although circuit theoretic abstractions have been used for decades for the design of individual RF components like antennas\cite{chu1948physical, balanis2015antenna}, matching networks\cite{steer2019microwave}, and amplifiers\cite{gray2009analysis}, the application of circuit models for MIMO communication systems is more recent.
	Phenomena like super-directivity\cite{9048753}  and super-bandwidth\cite{ akrout2022super} that occur in  tightly coupled arrays can be explained with the circuit theory approach.
	Hardware effects like amplifier current constraints\cite{deshpande2023analysis} and matching network limitations\cite{taluja_diversity_2013} can also be incorporated through circuit models.
	Hence, the circuit theory approach to communication is useful to design matching networks for optimizing achievable rate.

	Prior work has studied achievable rate maximization through impedance-matching only for specific matching network topologies\cite{saab_optimizing_2022,kundu_enhancing_2017,shyianov_achievable_2022,saab_capacity_2019,saab_capacity_2019-1}.
	In \cite{kundu_enhancing_2017},  an upper bound on the MIMO-OFDM capacity was proposed by optimizing the receiver matching network parameters based on a T-network topology.	 
	 In \cite{saab_optimizing_2022}, the achievable rate of a MISO and SIMO system was optimized in terms of the  inductances and transformer turns ratio of a single port  matching network.
	 Although \cite{saab_optimizing_2022} and \cite{kundu_enhancing_2017} used a communication theoretic objective, the methods used for optimizing the matching network parameters are specific to a given topology and do not guarantee optimality over a general family of passive and linear matching networks.
	From a  circuit theory perspective, there exists a fundamental limit on the wideband performance of  a passive matching network, popularly known as the Bode-Fano limits\cite{steer2019microwave, bode1945network,fano1950theoretical}.
Recent work derived an upper bound on the single-input-single-output (SISO) achievable rate by applying the Bode-Fano wideband matching constraints at the receiver\cite{shyianov_achievable_2022} and transmitter\cite{nitish_globecom23}.
	Recently, a multiport extension of the Bode-Fano matching  theory proposed new bounds applicable to a system with multiple transmit antennas driven by multiple sources\cite{nie_bandwidth_2017,nie_bandwidth_2017-1}.
	The application of the improved Bode-Fano matching  limits  to a MIMO system from an achievable rate perspective is  not investigated in prior work.
	
	In this paper, we analyze a MISO system from a joint circuit and communication theoretic perspective. 
	We answer two fundamental questions.
	The first question is ``What is the upper bound on the  achievable rate of a MISO system over all physically realizable linear and passive matching networks that satisfy the Bode-Fano wideband matching constraints?"
	We demonstrate how ignoring the Bode-Fano constraints 
	leads to an over-estimation of the rate for wideband systems.
	 The second question is ``How to design impedance-matching networks
	 that achieve rate close to the proposed upper bound?"
	 In contrast to prior work, we design realizable matching networks that maximize achievable rate.
	The main contributions of this paper are  as follows.
	\begin{itemize}
		\item We derive a frequency-selective circuit theoretic model of a MISO system with a single RF chain at the transmitter that supplies power to the antenna array through an impedance-matching network and an analog beamforming network. The system model formulation is in terms of the scattering parameters, which enables a direct application to any RF system whose scattering parameters can be measured. 
		\item We propose a     	
		general framework for optimizing the achievable rate of wideband MISO systems as a function of the impedance-matching network. The constraints are based on a  generalized version of the Bode-Fano wideband matching theory recently proposed in \cite{nie_bandwidth_2017,nie_bandwidth_2017-1}. 
		For deriving the constraints, we use a rational and passive approximation of the equivalent load comprising of the analog beamforming network and the transmit antennas.
		The maximum achievable rate is expressed in terms of the optimized transmission coefficient and Lagrangian parameters associated with the Bode-Fano inequalities.  The  transmission coefficient depends not only  on the antenna and array parameters but also on the wireless propagation channel and the analog beamformer.
		\item We propose a three step procedure to  design circuits that approximate the desired optimal response obtained through  the achievable rate optimization solution. Our simulation results show that the  matching network designed using this procedure achieves rates close to  the maximum achievable rate bound.
		\item We demonstrate this three step procedure for two specific models: a single Chu's electric antenna and an array of two Chu's antennas.  We use a practical LC ladder matching network topology whose components are numerically optimized in ADS to fit the corresponding optimal transmission coefficient.
		\item We compare our proposed bound and the performance of the designed matching network with the ideal Shannon's bound, frequency-flat transmission, conjugate matching at center frequency, and the no matching case.
		We also analyze the achievable rate trend with  bandwidth. We show the existence of an optimal bandwidth  for the achievable rate  bound obtained with Bode-Fano constraints and the corresponding circuit simulations.
	\end{itemize}
	
	This paper is organized as follows: In Section~\ref{System model}, we discuss the choice of the modeling framework in comparison with other frameworks used in literature.  Choosing a circuit theoretic modeling methodology, we formulate a frequency-selective model for a MISO system where the linear network parameters are described using the scattering parameter notation. In Section~\ref{sec opt ach rate bf con}, we discuss the general form of the Bode-Fano matching constraints followed by the achievable rate optimization problem formulation and derivation of the optimal  transmission coefficient. In Section~\ref{sec matching network design and illustrations}, we propose a methodology to design matching networks based on the derived transmission coefficient supported by  circuit illustrations  using  ADS software. In Section~\ref{sec: numerical results}, we present numerical results for SNR  and achievable rate using the derived theoretical bounds, circuit simulations, and comparison with conventional matching benchmarks. In Section~\ref{sec: conclusion}, we summarize the key takeaways  and discuss future research directions.
	 The simulation code  for generating achievable rate optimization results and the corresponding circuit ADS files are made publicly available to facilitate reproducibility\footnote{{https://github.com/nvdeshpa/AchievableRateWidebandMatching}}.
	
	\textit{Notation}:  A bold lowercase letter $\bm{\sfa}$ denotes a column vector, 
	a bold uppercase letter $\bm{\sfA}$ denotes a matrix,  $(\cdot)^{\ast}$ denotes conjugate,  $(\cdot)^T$ denotes transpose,   $|\cdot|$  indicates absolute value, $ \bI_N$ represents the identity matrix of size $N$, $\mathbf{0}_{N}$ represents an all zero matrix of size $N$,
	$\cR(z)$ denotes the real part of a complex number $z$, $ \{i\} _{1}^{N}$ is shorthand for $i = \{ 1, 2, \dots, N\}$, $[z]^{+}= \mbox{max}(0,z)$.
	

	\section{System model}\label{System model}
%
%

	\subsection{Modeling frameworks for wireless communication system}
	

\begin{figure*}[t!]
	\begin{subfigure}[t]{0.32\textwidth}   
		\centering 
		\includegraphics[width=0.875\textwidth]{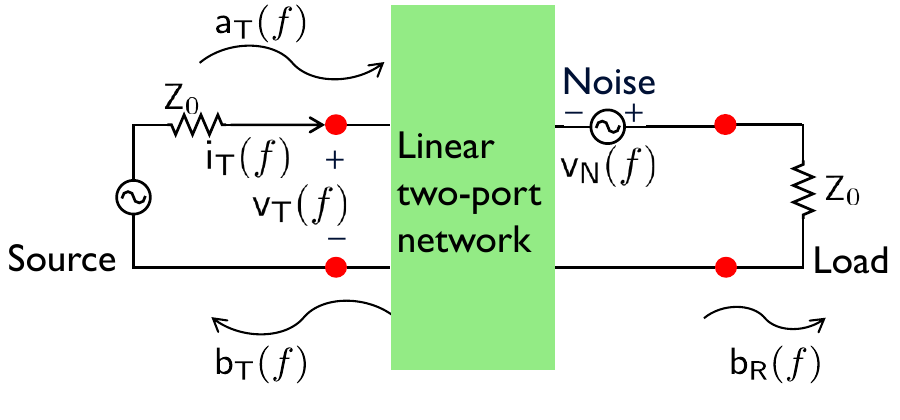}
		\caption
		{A general two-port network model where  $\mathsf{v}_{\mathsf{T}}(f)$ is the voltage, $\mathsf{i}_{\mathsf{T}}( f)$ is the current at the input, and $\mathsf{v}_{\mathsf{N}}( f)$ is noise source at output. Alternate representation using root power waves 
			$\mathsf{a}_{\mathsf{T}}( f)$, $\mathsf{b}_{\mathsf{T}}( f)$, and $\mathsf{b}_{\mathsf{R}}( f)$. } 
		\label{fig: 2 port network}
	\end{subfigure}
\hfill
	\begin{subfigure}[t]{0.65\textwidth}   
		\centering
		\includegraphics[width=\textwidth]{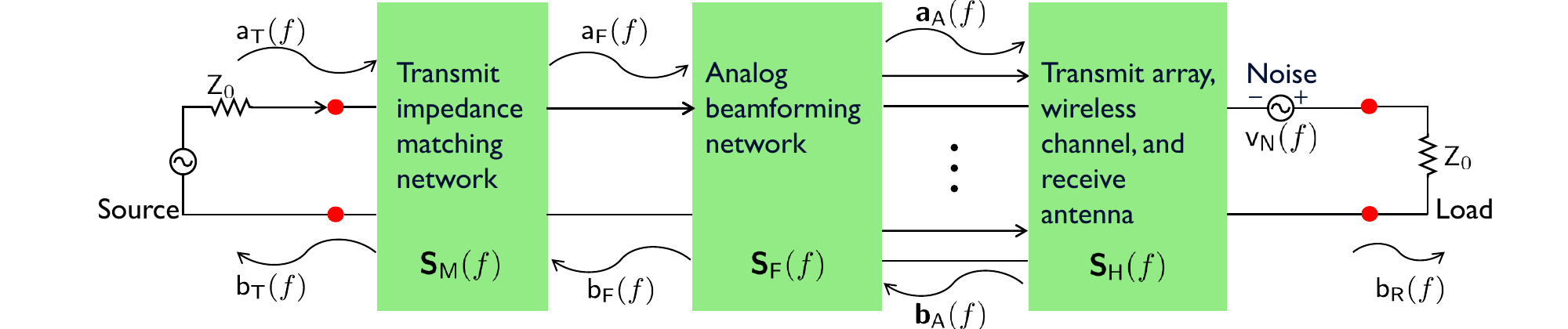}
		\caption
		{Circuit model of a MISO system where  $	\mathsf{a}_{\mathsf{F}}(f)$ is the incident root power wave and $	\mathsf{b}_{\mathsf{F}}(f)$ is the reflected root power wave from the analog beamforming network,  $\bm{\sfa}_{\mathsf{A}}( f) \in \bbC^{N\times 1}$ is the incident root power wave vector and  $\bm{\sfb}_{\mathsf{A}}( f) \in \bbC^{N\times 1}$ is the reflected root power wave vector from the transmit antenna array. The linear blocks are described using scattering parameter matrices $\bm{\sfS}_{\mathsf{M}}(f)$, $\bm{\sfS}_{\mathsf{F}}(f)$, and $\bm{\sfS}_{\mathsf{H}}(f)$.  }
		\label{fig:ckt model full} 
	\end{subfigure}
	\caption{In (a), a communication system with one RF chain at both receiver and transmitter is modeled as a two-port  network.  In (b),  the two-port network model is specified for a MISO system with one transmit RF chain connected to $N$ antennas through impedance-matching network and analog beamforming network. }
	\label{fig: ckt model}
\end{figure*}


With every wireless generation, there has been evolution in the communication system modeling approach. 
	The analysis of a wireless system is impacted by the choice of model.
It is important to choose a model that encompasses the proper assumptions and
	ensures the validity and applicability of the  insights to the target wireless application.

We overview different frameworks that are used for modeling  physical layer wireless communication.
	In terms of frequency dependence, the system model is classified as frequency-flat or frequency-selective.
	For frequency-flat models, the network parameters are evaluated at a specific center frequency and assumed to be fixed over the narrow bandwidth of interest.
	For  narrowband wireless applications like sensor networks, RFID\cite{decarli2018phase}, and narrowband Internet of Things\cite{9214385},  it suffices to use a frequency-flat modeling framework.
	For wideband wireless applications like satellite communication, Wi-Fi, and cellular\cite{6998944},    frequency-selective models are required to guarantee that the analytical or numerical results  from the model are useful for the desired frequency range.

	In terms of the modeling methodology, the models can be classified as   dimensionless  (non-circuit theoretic models) and  physically consistent  (circuit theoretic models).
	Non-circuit theoretic models have been useful for analyzing achievable rate, interference analysis, power allocation\cite{heath_jr_foundations_2018}, and beamforming optimization.
	However, the definition of power with these models is based on a single complex dimensionless variable.
The power definition using circuit models  is based on a pair of variables like the current and voltage or incident and reflected root power waves\cite{ivrlac_toward_2010},\cite{steer2019microwave}.
	For problems like impedance-matching network design which relate to power flow\cite{saab_optimizing_2022}, or analysis of new array architectures like dynamic metasurface antennas\cite{9762020},
	 it is essential to use a circuit model to capture the relevant hardware and electromagnetic effects like mutual coupling\cite{wallace_mutual_2004} and polarization\cite{9723156}.
	

	The circuit theoretic MIMO models can be further classified based on  impedance/admittance versus scattering parameters. Although impedance and scattering parameters can be converted to each other through algebraic transformations, the scattering parameters are more applicable because they can be easily measured for any general load and directly relate to the flow of power\cite{steer2019microwave}. Scattering parameters are widely recognized in the RF  community for design of individual RF components like antennas\cite{balanis2015antenna}, matching networks\cite{steer2019microwave}, and amplifiers\cite{gray2009analysis}.
	The use of scattering parameters for the analysis of wireless communication systems allows us to  leverage several results on matching network analysis developed in the microwave and antenna community\cite{nie_bandwidth_2017,nie_bandwidth_2017-1}. 
	It also makes our work generalizable to any practical RF system. 
As we target the achievable rate analysis and matching network design problem for wideband systems, 
	we use the circuit theoretic frequency-selective model with scattering parameters.
	
%
%

	\subsection{ A two-port linear network model of communication system}\label{subsec 2 port linear}

	In Fig.~\ref{fig: ckt model}(a), we represent a communication system with a single source and single load using a two-port network model\cite{shyianov_achievable_2022}. On the transmitter side, the source  generates the transmit signal obtained from the  output of the transmit RF chain, i.e., the signal obtained after   up-conversion and amplification.  This signal is input to a cascade of linear networks effectively modeled as a 	
	linear two-port network. The two-port network is used to model several linear blocks of a communication system like the 
	impedance-matching network, analog beamforming network, transmit antenna array network, wireless propagation channel, and receive antenna as shown in Fig.~\ref{fig: ckt model}(b) and described in Section~\ref{subsec ckt model MISO}. 
	We model the noise contribution  from background radiation at the receive antenna by  a voltage source at the output of the linear two-port network\cite{ivrlac_toward_2010}.  The 
	receiver RF chain is  modeled as a load. For simplifying the analysis, we do not model the low noise amplifier, receive matching network, and intrinsic noise source.
	

	We use a frequency domain representation for describing the signal flow through the two-port network. The subscript $``\mathsf{T}"$ indicates transmitted signal while subscript $``\mathsf{R}"$ indicates received signal.
	Let the voltage of the real-valued passband transmit signal in time domain at the input of the two-port network be $v_{\mathsf{T}}(t)$. Let the current entering the two-port network be $i_{\mathsf{T}}(t)$.
	We assume that the stochastic signals $v_{\mathsf{T}}(t)$ and $i_{\mathsf{T}}(t)$ are  Gaussian wide-sense stationary random processes so that these signals are completely described by their mean and second-order moments\cite{davenport1958introduction}.
	As these signals do not have finite {energy},   a windowed Fourier transform with interval $T_0$ is used for defining the spectrum\cite{russer2014modeling}.
	For frequency domain representation, we use frequency $f$ in Hertz.
	Let  $\mathsf{v}_{\mathsf{T}}(f)= \int_{-\frac{T_0}{2}}^{\frac{T_0}{2}} v_{\mathsf{T}}(t) e^{-\sfj 2 \pi  f t} \mathrm{d}t \big[\frac{\text{V}}{\text{Hz}}\big]$ and 
	$\mathsf{i}_{\mathsf{T}}( f)= \int_{-\frac{T_0}{2}}^{\frac{T_0}{2}} i_{\mathsf{T}}(t) e^{-\sfj 2 \pi  f t} \mathrm{d}t \big[\frac{\text{A}}{\text{Hz}}\big]$.
	The main purpose of using the frequency domain representation is to simplify the analysis.

	For further simplification, we use a root power wave representation of the signals which  directly relates to the flow of power\cite{steer2019microwave}.
	The root power waves  at different ports are related by the scattering parameters, which can be easily measured compared to impedance or admittance parameters.
	 The root power wave is  a stochastic process which can be expressed as a linear combination of the current and voltage stochastic processes.
	In terms of the voltage and current, assuming characteristic impedance of the transmit circuit as  $\sfZ_0$, the root power wave incident on the two-port network is defined as  $\mathsf{a}_{\mathsf{T}}( f)= \frac{\mathsf{v}_{\mathsf{T}}( f)+ \sfZ_0 \mathsf{i}_{\mathsf{T}}( f)}{2 \sqrt{\cR \{\sfZ_0\}}} \big[\frac{\sqrt{\text{W}}}{\text{Hz}}\big]$ \cite[Eq 2.118]{steer2019microwave}.  The root power wave reflected back from the two-port network on the transmit side is  defined as  $\mathsf{b}_{\mathsf{T}}( f)= \frac{\mathsf{v}_{\mathsf{T}}( f)- \sfZ_0^{\ast} \mathsf{i}_{\mathsf{T}}( f)}{2 \sqrt{\cR \{\sfZ_0\}}} \big[\frac{\sqrt{\text{W}}}{\text{Hz}}\big]$ \cite[Eq 2.118]{steer2019microwave}. 
	The power spectral density of the power incident (or available) on the transmitter side of the two-port network is  \cite{steer2019microwave,shyianov_achievable_2022}
	\begin{equation}\label{eqn Pt j omega}
		P_{\mathsf{T}}(f) = \underset{T_0 \rightarrow \infty}{\mbox{lim}}\frac{1}{T_0} \bbE[ |\mathsf{a}_{\mathsf{T}}(f)|^2] \left[\frac{\text{W}}{\text{Hz}}\right].
	\end{equation}
	Although currents and voltage signals can be used for formulating the communication system model, 
	the root power wave notation is a good mathematical tool for simplifying the impedance-matching problem in terms of metrics like power loss ratio and transmission coefficient\cite{6668983}.
	
	On the receiver side, we assume an ideal load termination  $\sfZ_0$ to  avoid reflected root power wave from the load and simplify the analysis.  This load termination requires the use of a  receive matching network that transfers all incident power to the receive RF chain. The modeling of a  practical receive matching network was done in \cite{shyianov_achievable_2022}  while we  focus only on the practical transmit matching network analysis and design.
	The root power wave 
	at the output of the linear two-port network 
	represents the received signal component  $\mathsf{b}_{\mathsf{RS}}( f)$.
	The voltage source at the receiver port models the noise from background radiation.
	 The noise voltage source is denoted as $\mathsf{v}_{\mathsf{N}}( f)$.
	 The resulting root power wave is  $\mathsf{b}_{\mathsf{RN}}( f)= \frac{\mathsf{v}_{\mathsf{N}}( f)}{\sqrt{\cR \{\sfZ_0\}}}$\cite{wallace_mutual_2004}.
	Adding the signal and noise root power waves, the resulting 
	root power wave incident on the load is denoted as $\mathsf{b}_{\mathsf{R}}(f)$\cite{wallace_mutual_2004}.
	By replacing $\mathsf{a}_{\mathsf{T}}(f)$ with $\mathsf{b}_{\mathsf{RS}}( f)$ in \eqref{eqn Pt j omega},  we obtain  the received signal power spectral density  $	P_{\mathsf{RS}}(f)$.  Similarly, by replacing $\mathsf{a}_{\mathsf{T}}(f)$ with  $\mathsf{b}_{\mathsf{RN}}( f)$ in \eqref{eqn Pt j omega}, we obtain the received noise spectral density  $	P_{\mathsf{RN}}(f)$.
	Let  $k_B$ be  the Boltzmann constant in J/K and  $T$ be temperature in K. We set $	P_{\mathsf{RN}}(f) = \mathsf{N}_0= k_B T \left[\frac{\text{W}}{\text{Hz}}\right]$.
	At the receiver, we define $\mathsf{SNR}(f)= \frac{P_{\mathsf{RS}}(f)}{\mathsf{N}_0}$.
	The $\mathsf{SNR}(f)$ is non-zero in the band for which  $P_{\mathsf{T}}(f)$  is  non-zero.
	
	Assuming Gaussian channel noise, the mutual information per unit time (bits/s) between the transmit and  received Gaussian random process is   $\int_{-\infty}^{\infty}\frac{1}{2}\log_2(1+ \mathsf{SNR}( f) ) \mathrm{d}f $. For a real-valued passband signal, $\mathsf{SNR}( f) $ is symmetric about $f=0$. This simplifies the definition to an integral over positive frequencies as $\int_{0}^{\infty}\log_2(1+ \mathsf{SNR}( f) ) \mathrm{d}f $\cite{shyianov_achievable_2022}.
	Although integration upper limit is unbounded,  SNR is positive only for a specific frequency range corresponding to the signal bandwidth which results in a finite integral value.
	In Section~\ref{subsec ckt model MISO}, we describe the  model of a MISO communication system and define  $\mathsf{SNR}(f)$ in terms of $	P_{\mathsf{T}}(f) $ and the scattering parameters of the individual linear sub-networks.
	
	\subsection{Circuit model of a MISO communication system in terms of scattering parameters}\label{subsec ckt model MISO}

	We analyze the achievable rate of a MISO wireless system consisting of a transmit array with $N$ antennas and a single receive antenna.
	The transmitter hardware consists of a single RF chain connected to $N$ antennas through an impedance-matching network and an analog beamforming network.
	The proposed model is applicable to any general type of antenna, array, analog beamformer, and matching network.  In this paper, we focus on the design of transmit matching network.  
	
	The circuit theoretic model of the MISO communication system is shown in Fig.~\ref{fig: ckt model}(b).
	The   transmit array with $N$ antennas and the  single receive antenna form an $(N+1)$ port network.  
	Let $\bm{\sfS}_{\mathsf{T}}( f) \in \bbC^{N\times N}$ be the scattering parameter matrix for the transmit array and $\sfS_{\mathsf{R}}( f) \in \bbC$ be the scattering parameter of the receive antenna. The wireless propagation channel scattering parameter is  $\bm{\sfs}_{\mathsf{RT}}^T( f) \in \bbC^{1\times N}$, which accounts for the  antenna gains and frequency-selective fading between receiver and transmitter. 
	Similarly, $\bm{\sfs}_{\mathsf{TR}}( f) \in \bbC^{N\times 1} $\cite{ivrlac_toward_2010}.
		%
		We assume that the transmit array is sufficiently far from the receive antenna such that the signal attenuation between them is large\cite{ivrlac_toward_2010,wallace_mutual_2004}. Hence, we can use the unilateral approximation by setting $\bm{\sfs}_{\mathsf{TR}}( f) = \mathbf{0}_{N\times 1}$, i.e., we assume that the transmitter is unaffected by the electromagnetic fields at the receiver.
		The $(N+1)$ port network block scattering parameter matrix is defined and simplified as
	\begin{equation}\label{eqn: SHf}
		\bm{\sfS}_{\mathsf{H}}( f) =  \begin{bmatrix}
			\bm{\sfS}_{\mathsf{T}}(f) &  \bm{\sfs}_{\mathsf{TR}}( f)\\
			\bm{\sfs}_{\mathsf{RT}}^T(f) & \sfS_{\mathsf{R}}( f)\\
		\end{bmatrix}= \begin{bmatrix}
			\bm{\sfS}_{\mathsf{T}}( f) &  \mathbf{0}_{N\times 1} \\
			\bm{\sfs}_{\mathsf{RT}}^T( f) & \sfS_{\mathsf{R}}( f)\\
		\end{bmatrix}.
	\end{equation}
	The incident  root power wave vector on the transmit antenna array is denoted as $\bm{\sfa}_{\mathsf{A}}( f) \in \bbC^{N\times 1}$  and the reflected root power wave vector is $\bm{\sfb}_{\mathsf{A}}( f)  \in \bbC^{N\times 1}$, as shown in Fig.~\ref{fig: ckt model}(b). 
		The root power wave vectors at the receiver and transmitter ports are related using $		\bm{\sfS}_{\mathsf{H}}( f)$ as 
		\begin{equation}\label{eqn: rx tx port relation}
		[\bm{\sfb}^T_{\mathsf{A}}( f), 	\mathsf{b}_{\mathsf{RS}}( f) ]^T=  \bm{\sfS}_{\mathsf{H}}(f) [	\bm{\sfa}^T_{\mathsf{A}}( f), 0]^T.
		\end{equation}
		With the unilateral approximation, we can isolate the transmitter circuit model 
		by writing $\bm{\sfb}_{\mathsf{A}}( f)=  	\bm{\sfS}_{\mathsf{T}}( f) \bm{\sfa}_{\mathsf{A}}( f) $.
Using \eqref{eqn: SHf} and \eqref{eqn: rx tx port relation}, the received signal root power wave is
		\begin{equation}\label{eqn b RS}
			\mathsf{b}_{\mathsf{RS}}( f)= \bm{\sfs}_{\mathsf{RT}}^T( f)\bm{\sfa}_{\mathsf{A}}( f).
		\end{equation}
		The unilateral approximation  enables simplification of the communication system analysis and design. It is also reasonable from a practical perspective because the signal attenuates heavily from transmitter to the receiver~\cite{ivrlac_toward_2010}.
		
		The transmitter network is characterized through the transmit impedance-matching network and analog beamforming network.		
		The scattering parameter matrix of the analog beamforming network is denoted as the $(N+1)\times (N+1)$ complex matrix $	\bm{\sfS}_{\mathsf{F}}(f) =\begin{bmatrix}
			\sfS_{{\mathsf{F}},11}(f) & \bm{\sfs}_{{\mathsf{F}},12}^T(f)\\
			\bm{\sfs}_{{\mathsf{F}},21}(f) & \bm{\sfS}_{{\mathsf{F}},22}(f)\\
		\end{bmatrix}$, where $	\sfS_{{\mathsf{F}},11}(f)\in \bbC$, $\bm{\sfs}_{{\mathsf{F}},12}^T(f) \in \bbC^{1\times N}$, $	\bm{\sfs}_{{\mathsf{F}},21}(f)  \in \bbC^{N\times 1} $, and $\bm{\sfS}_{{\mathsf{F}},22}(f) \in \bbC^{N\times N}$. 
		The combination of the antenna array and the analog beamforming network can be treated as an equivalent load with scattering parameter denoted as $\sfS_{\mathsf{eq}}(f)$. We express $\sfS_{\mathsf{eq}}(f)$ in terms of the scattering parameter matrix elements of the array and the analog network as \cite{nie_bandwidth_2017}
		\begin{align}\label{eqn S eq}
			\sfS_{\mathsf{eq}}(f) &= 	\sfS_{{\mathsf{F}},11}(f)  + \bm{\sfs}_{{\mathsf{F}},12}^T(f) 	\bm{\sfS}_{\mathsf{T}}(f) \\ \nonumber &\times (\bI-\bm{\sfS}_{{\mathsf{F}},22}(f)	\bm{\sfS}_{\mathsf{T}}(f) )^{-1}	\bm{\sfs}_{{\mathsf{F}},21}(f).
		\end{align}
		The scattering parameter matrix of the transmit impedance-matching network, which connects the transmit source to the equivalent load,  is denoted as $\bm{\sfS}_{\mathsf{M}}(f) =\begin{bmatrix}
			\sfS_{{\mathsf{M}},11}(f) & 	\sfS_{{\mathsf{M}},12}(f)\\
			\sfS_{{\mathsf{M}},21}(f)& 	\sfS_{{\mathsf{M}},22}(f)\\
		\end{bmatrix} $. For the single antenna case, the transmitter network only consists of the matching network. As there is no analog beamformer for the single antenna, the scattering parameter of the equivalent load is the scattering parameter of the antenna, i.e., $	\sfS_{\mathsf{eq}}(f) = \sfS_{\mathsf{T}}(f) $.
	
	To establish a linear relationship between the received signal root power wave $	\mathsf{b}_{\mathsf{RS}}(f)$ and the transmit signal  root power wave $\mathsf{a}_{\mathsf{T}}(f)$, we apply the scattering parameter definition for each block shown in Fig.~\ref{fig: ckt model}(b).
		Let the incident root power wave on the combined load of antennas and analog beamforming network be denoted as $	\mathsf{a}_{\mathsf{F}}(f)$. The reflected root power wave from  the combined load is $\mathsf{b}_{\mathsf{F}}(f) =\sfS_{\mathsf{eq}}(f) 	\mathsf{a}_{\mathsf{F}}(f) $.
		We express $\bm{\sfa}_{\mathsf{A}}(f) $ in terms of $\mathsf{a}_{\mathsf{F}}(f)$ using the scattering parameter matrix elements as 
		\begin{equation}\label{eqn  aA aF}
			\bm{\sfa}_{\mathsf{A}}(f) =(\bI-\bm{\sfS}_{{\mathsf{F}},22}(f)	\bm{\sfS}_{\mathsf{T}}(f) )^{-1}	\bm{\sfs}_{{\mathsf{F}},21}(f)	\mathsf{a}_{\mathsf{F}}(f).
		\end{equation}
		Finally,  $	\mathsf{a}_{\mathsf{F}}(f)$  is expressed in terms of the incident root power wave on the transmit impedance-matching network 
		\begin{equation}\label{eqn  aF aM}
			\mathsf{a}_{\mathsf{F}}(f) = \frac{\sfS_{{\mathsf{M}},21}(f)}{1- \sfS_{{\mathsf{M}},22}(f)	\sfS_{\mathsf{eq}}(f)}	\mathsf{a}_{\mathsf{T}}(f).
		\end{equation}
		Using \eqref{eqn b RS}, \eqref{eqn S eq}, \eqref{eqn  aA aF}, \eqref{eqn  aF aM}, we relate $\mathsf{b}_{\mathsf{RS}}(f)$  to $\mathsf{a}_{\mathsf{T}}(f)$ using a channel coefficient corresponding to an equivalent SISO channel 
		\begin{equation}\label{eqn: b RS in terms of aT}
	H_{\mathsf{SISO}}(f)=	\frac{ \bm{\sfs}_{\mathsf{RT}}^T(f)  (\bI-\bm{\sfS}_{{\mathsf{F}},22}(f)	\bm{\sfS}_{\mathsf{T}}(f) )^{-1}	\bm{\sfs}_{{\mathsf{F}},21}(f) \sfS_{{\mathsf{M}},21}(f)}{1- \sfS_{{\mathsf{M}},22}(f)	\sfS_{\mathsf{eq}}(f)}	.
		\end{equation}
	Hence, we have 	$	\mathsf{b}_{\mathsf{RS}}(f)= H_{\mathsf{SISO}}(f)\mathsf{a}_{\mathsf{T}}(f)$. The equivalent channel in \eqref{eqn: b RS in terms of aT} depends not only on the wireless propagation channel but also captures the  frequency-selectivity effect of antennas and matching network.

		The equivalent channel expression is used for relating the received signal power spectral density to that available from the transmitter side.
		Using \eqref{eqn: b RS in terms of aT} and the definition of the power spectral density in \eqref{eqn Pt j omega}, we relate $	P_{\mathsf{RS}}(f)$ to $	P_{\mathsf{T}}(f)$ as 
		\begin{equation}
			P_{\mathsf{RS}}(f) = |H_{\mathsf{SISO}}(f)|^2 P_{\mathsf{T}}(f)\left[\frac{\text{W}}{\text{Hz}}\right].
		\end{equation}
		The  $\mathsf{SNR}(f)$ in terms of $	P_{\mathsf{T}}(f) $ and the equivalent channel is 
		\begin{equation}\label{eqn: SNR omega}
			\mathsf{SNR}(f)= |H_{\mathsf{SISO}}(f)|^2\frac{ P_{\mathsf{T}}(f)}{\mathsf{N}_0}.
		\end{equation}
		The mutual information per unit time is $\int_0^{\infty}\log_2\left(1+ |H_{\mathsf{SISO}}(f)|^2\frac{ P_{\mathsf{T}}(f)}{\mathsf{N}_0} \right) \mathrm{d}f$.
		It depends on the  design of the matching network and the transmit power allocation at each frequency.
		We assume a bandlimited source that supplies a maximum power per frequency $\mathsf{E}_{\mathsf{s}}\left[\frac{\text{W}}{\text{Hz}}\right]$  for $f\in [f_{\mathsf{min}}, f_{\mathsf{max}}]$\cite{shyianov_achievable_2022}.  
		 For the bandwidth $B =f_{\mathsf{max}} - f_{\mathsf{min}}$,  we assumed a fixed total supplied power $B  \mathsf{E}_{\mathsf{s}} \text{W}$.		
		Assuming  that the source supplies the maximum available power at each frequency, 
		we define the achievable rate  in bits/s as 
		\begin{equation}\label{eqn: achievable rate def}
			\mathsf{R}=\int_{f_{\mathsf{min}}}^{f_{\mathsf{max}}} \log_2\left(1+ |H_{\mathsf{SISO}}(f)|^2\frac{ \mathsf{E}_{\mathsf{s}}}{\mathsf{N}_0} \right) \mathrm{d}f.
		\end{equation}
		In Section~\ref{sec opt ach rate bf con}, we optimize the achievable rate $	\mathsf{R}$ by optimally designing the matching network.

					\section{Optimizing achievable rate under Bode-Fano matching constraints}
					\label{sec opt ach rate bf con}

					In this section, we propose a general framework for optimizing the achievable rate as a function of the matching network.  The rate depends on the matching network through the term $H_{\mathsf{SISO}}(f)$ as shown in \eqref{eqn: b RS in terms of aT}
					which depends on the matching network scattering parameter matrix elements $\sfS_{{\mathsf{M}},21}(f)$ and $\sfS_{{\mathsf{M}},22}(f)$. The values for these elements at each frequency are hardware specific, i.e., dependent on the actual physical elements like the inductors and capacitors used for building the matching network circuit. One approach of optimizing rate is to define the rate objective in terms of the  physical component values for a fixed structure\cite{saab_optimizing_2022}. A limitation of this approach is that it does not guarantee that the specific circuit gives better performance theoretically  than any other physically realizable matching network. In this paper, instead of optimizing the rate in terms of the physical components of a specific matching circuit, we formulate and solve a general problem that applies to any  matching network made of passive and linear elements.  In Section~\ref{subsec general form of BF MC}, we describe the constraints associated with any general  passive matching network.
					
					%
					%

					\subsection{General form of the Bode-Fano matching constraint}\label{subsec general form of BF MC}
					
					The first systematic approach to study the bandwidth limitation of a matching network was proposed by Bode  for a special type of reactive load\cite{bode1945network}.  The work by Bode was generalized by Fano  for arbitrary reactive loads\cite{fano1950theoretical}. These results on the bandwidth limitation of matching networks are popularly known as Bode-Fano  limits\cite{steer2019microwave}. 
					Recently, a generalization of the Bode-Fano matching  limits was proposed  for arbitrary loads like an antenna array with an analog beamforming network\cite{nie_bandwidth_2017,nie_bandwidth_2017-1}.  We use these results from \cite{nie_bandwidth_2017,nie_bandwidth_2017-1} to formulate the  matching network constraints in the achievable rate optimization problem.
					
					The Bode-Fano constraints place a bound on the power loss ratio metric~\cite{nie_bandwidth_2017}. 
				The constraints are expressed in terms of the power loss ratio, which indicates the ratio of	  expected power lost (due to reflection and dissipation) to the expected input power  where the expectation is over the input random signal.
					Mathematically, it is given as 
					\begin{align}\label{eqn:  power loss ratio}
						r^2(f)&=  1- \frac{\bbE[  |\sfa_{\mathsf{F}}(f)|^2-  |\sfb_{\mathsf{F}}(f)|^2 ]}{\bbE[|\sfa_{\mathsf{T}}(f)|^2]}\\ \nonumber &\stackrel{(a)}{=} 1-\frac{|\sfS_{{\mathsf{M}},21}(f) |^2}{|1- \sfS_{{\mathsf{M}},22}(f)	\sfS_{\mathsf{eq}}(f)  |^2}(1 - |	\sfS_{\mathsf{eq}}(f)|^2),
					\end{align}
					where equality $(a)$ follows from  the definition of  $\mathsf{b}_{\mathsf{F}}(f) $ and \eqref{eqn  aF aM}.
					We also define the transmission coefficient as
					\begin{equation}\label{eqn:  tau omega}
						\cT(f)= 1- 	r^2(f).
					\end{equation} 
					A lower value of $	r^2(f)$ or a higher value of $\cT(f)$ for a specified bandwidth indicates a better power transfer to the equivalent load in the desired band. 


				Cauchy's integral relations in complex variable calculus can be applied to any linear circuit model for deriving Bode-Fano constraints\cite{fano1950theoretical}. Mathematically, it is convenient to analyze the circuit model  as a function of the  complex frequency $s=\sigma+ \sfj 2 \pi f$\cite{fano1950theoretical}\cite{nie_bandwidth_2017-1}.
					Similar to the original Bode-Fano constraints, we  assume that the load should be realizable by means of finite passive lumped elements\cite{fano1950theoretical}.
					Therefore,  we use a rational approximation of   $\sfS_{\mathsf{eq}}(f) $  defined in the whole complex plane and denoted as  $\hat{\sfS}_{\mathsf{eq}}(s) $\cite{nie_bandwidth_2017-1}.
					We assume that $\hat{\sfS}_{\mathsf{eq}}(s)$ should be in the rational form and satisfy the passivity condition\cite{nie_bandwidth_2017-1}.
					The guidelines for obtaining  $\hat{\sfS}_{\mathsf{eq}}(s) $ from   $\sfS_{\mathsf{eq}}(f) $ are discussed in detail in \cite{nie_bandwidth_2017-1}.
					Note that $\hat{\sfS}_{\mathsf{eq}}(s) $ is not unique and depends on the technique used to approximate $\sfS_{\mathsf{eq}}(f) $.  We 
					briefly  summarize the approximation techniques from \cite{nie_bandwidth_2017-1} in Appendix~\ref{app: Seq s}.
					
					The Bode-Fano theory provides a set of constraints on the power loss ratio $	r^2(f)$ for any passive and linear impedance-matching network terminated with a passive load realized using lumped elements.
					These constraints
					are expressed as bounds on the integral of logarithm of the power loss ratio \cite{nie_bandwidth_2017,taluja_diversity_2013,fano1950theoretical}.
					For a simple load of resistor $R$ and capacitor $C$ in parallel, there is only one Bode-Fano constraint expressed  as $\int_0^{\infty} \log\left(\frac{1}{	r^2(  f)}\right) \mathrm{d}f \leq \frac{1}{RC}$~\cite{steer2019microwave}. 
					For an arbitrary load,  the number of necessary constraints for the physical realizability of $	r^2(f)$ is determined using a  Darlington equivalent representation of the load.  
					From Darlington's theory, any physically realizable impedance is equivalent to the input impedance of a reactive two-port network terminated with a $1\Omega$  resistor\cite{darlington1939synthesis}. 
					The number of such necessary constraints on $	r^2(f)$ equals  the number of independent parameters used to define the Darlington equivalent network of the load\cite{fano1950theoretical}. 
For example, for a load of resistor, inductor, and capacitor in series, the Darlington equivalent is specified using  the quality factor value and the resonant frequency which results in 
two Bode-Fano bounds\cite[Eq 11]{taluja_diversity_2013}.  

					For an equivalent load with rational approximation $\hat{\sfS}_{\mathsf{eq}}(s) $, we assume there are $N_{\mathsf{BF}}$ number of distinct constraints for describing the physical realizability of the power loss ratio. 
					For the $i$th constraint where $i\in \{1, 2, \dots, N_{\mathsf{BF}}\}$,  we define two positive terms $\xi_{\mathsf{BF}, i}( f)$ and  $B_{\mathsf{BF},i}$.
					The term $\xi_{\mathsf{BF}, i}(f)$  is a prelog term in the integrand which is multiplied by the 
					logarithm of the power loss ratio. The term $B_{\mathsf{BF},i}$ is an upper bound on the Bode-Fano integral.  
						For a load whose scattering parameter after rational approximation is $\hat{\sfS}_{\mathsf{eq}}(s) $, the $N_{\mathsf{BF}}$ distinct constraints required for the realizability of  power loss ratio  $r^2( f)$ are\cite{nie_bandwidth_2017}
					\begin{equation}\label{eqn: N BF constraints}
						\int_0^{\infty}\!\!\! \xi_{\mathsf{BF}, i}(  f) \log\left(\frac{1}{	r^2(  f)}\right) \mathrm{d}f \leq B_{\mathsf{BF}, i}, \text{for } \{i\} _{1}^{N_{\mathsf{BF}}} ,
					\end{equation}
					where  $\xi_{\mathsf{BF}, i}( f)$  and $B_{\mathsf{BF},i}$ are positive terms evaluated using   $\hat{\sfS}_{\mathsf{eq}}(s) $  as discussed in Appendix~\ref{app: fbfi bfi}.
				The detailed proof of the general form of the Bode-Fano constraint is given in \cite{nie_bandwidth_2017-1}. 
					
					In Section~\ref{subsec rate opt formulation}, we will use the constraints defined in \eqref{eqn: N BF constraints} for formulating the achievable rate optimization problem as a function of the transmission coefficient  $\cT(f)$ defined in \eqref{eqn:  tau omega}. Before proceeding to the problem formulation, we first rewrite the achievable rate  in terms of the transmission coefficient so that both constraints and objective in the optimization can be expressed as function of the  variable 
					$\cT(f)$.
					
					
					
					\subsection{Achievable rate optimization problem formulation}\label{subsec rate opt formulation}
					
					The achievable rate metric depends on the matching network through the scattering parameters $\sfS_{{\mathsf{M}},21}(f)$ and $\sfS_{{\mathsf{M}},22}(f)$ as shown through  \eqref{eqn: b RS in terms of aT} and \eqref{eqn: achievable rate def}. 
					Using  \eqref{eqn: b RS in terms of aT}, \eqref{eqn: SNR omega},   \eqref{eqn:  power loss ratio},   and \eqref{eqn:  tau omega},  we express $\mathsf{SNR}(f)$ in terms of $\cT(f)$ as 
					\begin{align}\label{eqn: SNR omega in terms of tau}
						\mathsf{SNR}(f)&= |\bm{\sfs}_{\mathsf{RT}}^T(f)  (\bI-\bm{\sfS}_{{\mathsf{F}},22}(f)	\bm{\sfS}_{\mathsf{T}}(f) )^{-1}	\bm{\sfs}_{{\mathsf{F}},21}(f)|^2  \\ \nonumber &\times\frac{ P_{\mathsf{T}}(f) \cT(f)}{(1 - |	\sfS_{\mathsf{eq}}(f)|^2)\mathsf{N}_0}    .
					\end{align}
					From \eqref{eqn: SNR omega in terms of tau}, the  achievable rate expression is 
					\begin{align}\label{eqn R in terms of Tau}
						\mathsf{R}=&\int_{f_{\mathsf{min}}}^{f_{\mathsf{max}}}\!\!\!\! \log_2\bigg(1+ \frac{|\bm{\sfs}_{\mathsf{RT}}^T( f)  (\bI-\bm{\sfS}_{{\mathsf{F}},22}( f)	\bm{\sfS}_{\mathsf{T}}(f) )^{-1}	\bm{\sfs}_{{\mathsf{F}},21}( f)|^2}{(1 - |	\sfS_{\mathsf{eq}}( f)|^2)} \\ \nonumber&\times \frac{ \mathsf{E}_{\mathsf{s}}}{\mathsf{N}_0} \cT( f) \bigg) \mathrm{d}f.
					\end{align}
					In the ideal matching network case, i.e., no power loss due to reflection or dissipation,  $\cT(f)=1$, and the ideal lossless SNR is expressed as
					\begin{equation}\label{eqn: SNR ideal}
						\mathsf{SNR}_{\text{ideal}}(f)= \frac{|\bm{\sfs}_{\mathsf{RT}}^T(f)  (\bI-\bm{\sfS}_{{\mathsf{F}},22}(f)	\bm{\sfS}_{\mathsf{T}}(f) )^{-1}	\bm{\sfs}_{{\mathsf{F}},21}(f)|^2}{(1 - |	\sfS_{\mathsf{eq}}(f)|^2)}  \frac{ \mathsf{E}_{\mathsf{s}}}{\mathsf{N}_0} .
					\end{equation}
					The achievable rate for the ideal case is $\mathsf{R}_{\text{ideal}}=\int_{f_{\mathsf{min}}}^{f_{\mathsf{max}}}  \log_2\left(1+ \mathsf{SNR}_{\text{ideal}}( f) \right) \mathrm{d}f$. The ideal SNR depends on the wireless propagation channel, the scattering parameters of the antenna array, and the analog beamforming network but does not depend on the matching network.
					For a physically realizable matching network, $\cT(f)\leq 1$. Therefore, $\mathsf{R} \leq  \mathsf{R}_{\text{ideal}}$ meaning that the achievable rate is over-estimated when Bode-Fano constraints are disregarded.

					We formulate the achievable rate optimization problem to  optimally design the transmission coefficient $\cT(f)$.
					In \eqref{eqn R in terms of Tau}, we defined the optimization objective in terms of $\cT( f)$. Similarly, the Bode-Fano inequalities from \eqref{eqn: N BF constraints} can be expressed in terms of $\cT(  f)$ using  \eqref{eqn:  tau omega}.
					The achievable rate optimization problem  in terms  of   $\cT(f)$ using \eqref{eqn: N BF constraints},  \eqref{eqn R in terms of Tau}, and  \eqref{eqn: SNR ideal} is 
					\begin{subequations}\label{problem1}
						\begin{alignat}{3}
							&\mathsf{R}_{\text{max}}= \underset{ \cT(  f) }{\mbox{ max }}
							\int_{f_{\mathsf{min}}}^{f_{\mathsf{max}}} \log_2(1+ \mathsf{SNR}_{\text{ideal}}( f)  \cT( f)) \mathrm{d}f
							, \\
							&\text{s.t. }\!\!\!
							\int_0^{\infty}\!\!\! \xi_{\mathsf{BF},i}( f) \log\left(\frac{1}{	1- \cT(  f)}\right) \mathrm{d}f \leq B_{\mathsf{BF},i}	,  \text{for } \{i\} _{1}^{N_{\mathsf{BF}}}  \label{problem1_a} \\
							&0\leq \cT(f) \leq 1 .\label{problem1_b}
						\end{alignat}
					\end{subequations}
					The constraint in \eqref{problem1_b} follows from the definition of the transmission coefficient in \eqref{eqn:  tau omega}. 
Comparing this  formulation to \cite{shyianov_achievable_2022},
					the key difference is that $\cT(f)$ was defined on the receiver side between a single antenna and the low-noise amplifier in \cite{shyianov_achievable_2022}. This led to $\cT(f)$ appearing in both signal power and the extrinsic noise power in \cite[Eq 21]{shyianov_achievable_2022}.
					The problem formulation in our work expressed in \eqref{problem1} uses the $\cT(f)$ defined between the transmit RF chain and the equivalent load of  multiple transmit antennas and analog beamforming network. So $\cT(f)$ appears only in the signal power leading to an  optimal solution expression different from \cite[Eq 25]{shyianov_achievable_2022}.

					\subsection{Optimal transmission coefficient}

	We use  the
Lagrangian  to solve the optimization problem in \eqref{problem1}. 
					 The total number of constraints  in  \eqref{problem1_a} and \eqref{problem1_b} is $N_{\mathsf{BF}}+2$. For the $i$th constraint, we denote the Lagrangian parameter  as $\mu_i$. The Lagrangian is \cite{boyd2004convex}
					\begin{align}\label{eqn: lagrangian}
						&	\cL\left(\cT(  f), \mu_i |_{i=1}^{N_{\mathsf{BF}}+2}\right)=	- \int_{f_{\mathsf{min}}}^{f_{\mathsf{max}}} \log_2(1+ \mathsf{SNR}_{\text{ideal}}( f)  \cT(  f)) \mathrm{d}f \nonumber\\ &+\sum_{i=1}^{N_{\mathsf{BF}}} \mu_i \left(\int_0^{\infty} \xi_{\mathsf{BF},i}(  f) \log\left(\frac{1}{	1- \cT( f)}\right) \mathrm{d}f - B_{\mathsf{BF},i}  \right)\nonumber\\ &-\mu_{N_{\mathsf{BF}}+1} \cT( f) + \mu_{N_{\mathsf{BF}}+2} (\cT( f)-1).
					\end{align}
The solution to \eqref{problem1} is obtained after applying \textit{Karush-Kuhn-Tucker} (KKT) conditions in Appendix~\ref{app: proof thm2}\cite{shyianov_achievable_2022}. The maximum rate is defined in terms of the optimal transmission coefficient $\cT^{\star}( f)$ as 
\begin{equation}
\mathsf{R}_{\text{max}}= 
\int_{f_{\mathsf{min}}}^{f_{\mathsf{max}}} \log_2(1+ \mathsf{SNR}_{\text{ideal}}( f) \cT^{\star}( f)) \mathrm{d}f.
\end{equation}
The expression for $\cT^{\star}( f)$ is in terms of the optimal Lagrangian parameters  $\mu^{\star}_i |_{i=1}^{N_{\mathsf{BF}}} $ described as follows.
					\begin{theorem}\label{thm: t and mu}
The relationship between the variables $\cT^{\star}(f) $ and $\mu^{\star}_i |_{i=1}^{N_{\mathsf{BF}}} $ corresponding to the optimal solution of the
optimization problem in \eqref{problem1} is as follows.
						\begin{subequations}\label{eqn thm}
							\begin{alignat}{3}
								&\cT^{\star}( f) = \left[\frac{1- \ln 2 \frac{\sum_{i=1}^{{N_{\mathsf{BF}}}} \mu^{\star}_i \xi_{\mathsf{BF},i}( f)}{\mathsf{SNR}_{\text{ideal}}( f) } }{1+ \ln 2\sum_{i=1}^{N_{\mathsf{BF}}} \mu^{\star}_i \xi_{\mathsf{BF},i}(f)}\right]^{+}, \quad  \mu^{\star}_i |_{i=1}^{N_{\mathsf{BF}}} \geq 0, \label{eqn: thm1}\\
								& \mu_i^{\star}	\!\left(\int_0^{\infty}\!\!\! \xi_{\mathsf{BF},i}( f) \log\left(\frac{1}{	1- 	\cT^{\star}(  f) }\right) \mathrm{d}f   - B_{\mathsf{BF},i}\right)=0,  \{i\} _{1}^{N_{\mathsf{BF}}},\label{eqn: thm2}\\
								&\left(\int_0^{\infty} \!\!\!\xi_{\mathsf{BF},i}( f) \log\left(\frac{1}{	1- 	\cT^{\star}( f) }\right) \mathrm{d}f   - B_{\mathsf{BF},i}\right)\leq 0,\{i\} _{1}^{N_{\mathsf{BF}}}\label{eqn: thm3}.
							\end{alignat}
						\end{subequations}
						
					\end{theorem}
					\begin{proof}
						Refer to Appendix~\ref{app: proof thm2} for proof.
					\end{proof}
					
					The expression for the optimal transmission coefficient computed using \eqref{eqn thm} can be interpreted as waterfilling in the frequency domain which is a fundamental result by Shannon in information theory\cite{1697831}.  From \eqref{eqn: thm1}, we observe that $\cT^{\star}( f)$ is higher for frequencies with better  $\mathsf{SNR}_{\text{ideal}}( f)$.
					As $\mathsf{SNR}_{\text{ideal}}( f) $ is inversely proportional to bandwidth, the  peak gain in $\cT^{\star}( f)$ for higher bandwidths is lower and vice versa.  This shows that the  fundamental gain-bandwidth tradeoff of matching networks is captured in the expression of $\cT^{\star}( f)$ in  \eqref{eqn: thm1}.

					The variables $\cT^{\star}(f) $ and $\mu^{\star}_i |_{i=1}^{N_{\mathsf{BF}}} $ corresponding to the optimal solution of the
					optimization problem in \eqref{problem1} are tightly coupled in the equations \eqref{eqn: thm1}, \eqref{eqn: thm2}, and \eqref{eqn: thm3}.
					We use a numerical approach to compute a sub-optimal solution.
					We set all but one Lagrangian parameters to 0, apply a bisection search on the non-zero parameter till \eqref{eqn: thm2} is satisfied within a threshold, and repeat this process for all parameters to obtain the values of $\mu^{\star}_i |_{i=1}^{N_{\mathsf{BF}}} $ that maximize the rate.	This low-complexity numerical approach ensures that 	\eqref{eqn: thm2} is satisfied  for $N_{\mathsf{BF}}-1$ parameters and   within a specific tolerance for one parameter.			
					 The approximate solution for  $\cT^{\star}(f) $ is obtained by substituting the optimized values of $\mu^{\star}_i |_{i=1}^{N_{\mathsf{BF}}} $ in 
					  \eqref{eqn: thm1}.

					All passive and linear matching networks will provide an achievable rate less than the value of $\mathsf{R}_{\text{max}}$.
					This maximum achievable rate based on Bode-Fano bounds is more accurate than $\mathsf{R}_{\text{ideal}}$ which disregards the matching theory.					
					This rate $\mathsf{R}_{\text{max}}$ is a new benchmark for designing  matching networks instead of a metric like power transfer efficiency which does not capture the effect of wireless propagation channel.

					

					\section{Matching network circuit design methodology and illustrations }\label{sec matching network design and illustrations}
					
					From a system design perspective, it is crucial to provide a practical methodology to approximate the theoretical achievable rate bound from Section~\ref{sec opt ach rate bf con}. 
					In this section, we  address the second question, ``How to design impedance-matching networks to approximate this achievable rate bound?"
					We propose a practical matching network design approach using $\cT^{\star}( f ) $ from \eqref{eqn: thm1}.
					
					\subsection{General methodology to design matching network }\label{subsec: general methodology}
					
					We provide a three step procedure to design a matching network based on the  achievable rate upper bound as follows.
					\begin{enumerate}
						\item \textbf{Evaluation of Bode-Fano constraints }
						\begin{enumerate}\label{step: 1a}
							\item For a given scattering matrix of an antenna $	\bm{\sfS}_{\mathsf{T}}( f)$ and analog beamforming network $	\bm{\sfS}_{\mathsf{F}}(f)$, obtain a passive rational approximation as a function of the complex frequency to evaluate the rational function of the scalar equivalent load $\hat{\sfS}_{\mathsf{eq}}(s)$.
							\item
							Evaluate $N_{\mathsf{BF}}$  Bode-Fano constraints  using the expression of $\hat{\sfS}_{\mathsf{eq}}(s)$ based on Table~\ref{tab: eval bf functions} in Appendix~\ref{app: fbfi bfi}.		
						\end{enumerate}
						\item\textbf{Optimal transmission coefficient}
						\begin{enumerate}
							\item  Solve the optimization problem \eqref{problem1} for the $N_{\mathsf{BF}}$  Bode-Fano constraints by numerically solving the system of equations and inequalities given by \eqref{eqn thm}.
							\item Compute  $\cT^{\star}( f ) $ for the optimized $\mu^{\star}_i |_{i=1}^{N_{\mathsf{BF}}} $ using \eqref{eqn: thm1}.
						\end{enumerate}
						\item \textbf{Approximating $\cT^{\star}( f ) $ with a practical matching network topology}
						\begin{enumerate}
							\item Choose a general reactive ladder circuit with a fixed order.
							\item Optimize the component
							values of the matching network circuit topology to fit the
							desired frequency response of the optimal transmission
							coefficient $\cT^{\star}( f ) $.
						\end{enumerate}
						
					\end{enumerate}

					In the first step, 	overfitting when approximating $\sfS_{\mathsf{eq}}(f) $ with $\hat{\sfS}_{\mathsf{eq}}(s) $ can result in loose Bode-Fano bounds.
					Sometimes, there exists poles and zeros in the rational approximation which are  close to each other.
					As observed from expression of $B_{\mathsf{BF}}$  from Table~\ref{tab: eval bf functions} in Appendix~\ref{app: fbfi bfi}, overfitting may result in higher bounding values of $B_{\mathsf{BF}}$\cite{nie_bandwidth_2017-1}.  
					Overfitting is an issue if the computed value of  $\cT^{\star}( f )$ is close to one even for higher bandwidths. The gain-bandwidth tradeoff will not be captured due to overfitting.
					This issue can be avoided by reducing the model order.

						\begin{figure}[htbp]
						\centering
						\includegraphics[width=0.5\textwidth]{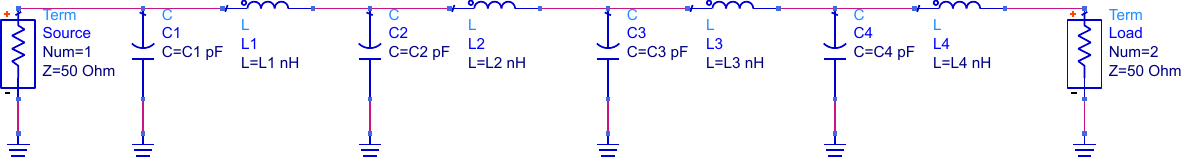}
						\caption{ Circuit model of the 4th order LC ladder  in ADS software.  This simple topology can be  used to approximate the transmission coefficient response obtained from the optimization problem.}
						\label{fig: SISO schematic order 4 MN}
					\end{figure}

					\subsection{Application of matching network design methodology to a single Chu's antenna }\label{subsec: single dipole}
					
					We present numerical illustrations for the matching network design methodology.
					For simulations, we use  Chu's antenna at  receiver and transmitter, similar to \cite{shyianov_achievable_2022}, to  provide generic insights without using a specific antenna design.
					Although Bode-Fano bounds depend on the antenna type, the achievable rate optimization methodology is general and can be applied to any antenna design.					
					
				\noindent \textbf{Chu's antenna model}: 					
					Let the Chu's antenna structure be enclosed in a spherical volume with radius $a$. Let the speed of light be denoted as $c$ and characteristic resistance  be denoted as $R$. The Chu's electric antenna is modeled with an equivalent circuit consisting of a capacitor $C= \frac{a}{c R}$ in series with a parallel combination of an inductor $L=\frac{a R}{c}$ and resistor $R$ \cite[Fig. 2]{shyianov_achievable_2022}.
					The input impedance is $\sfZ_{\mathsf{T}}(f)= \sfZ_{\mathsf{R}}(f)= \frac{R}{\sfj 2 \pi  f \frac{a}{c}}+ \frac{R}{1+({\sfj 2 \pi  f \frac{a}{c}})^{-1}}$. The scattering parameter in rational form is $\hat{\sfS}_{\mathsf{T}}(s)= ({2 s^2 \frac{a^2}{c^2}+ 2 s \frac{a}{c}+1})^{-1}.$
					For a single  antenna, there is no analog beamforming network, hence, $	\hat{\sfS}_{\mathsf{eq}}(s) = 	\hat{\sfS}_{\mathsf{T}}(s)$.
					
				\noindent	\textbf{Bode-Fano bounds for a single Chu's antenna}:		
					Substituting $\hat{\sfS}_{\mathsf{eq}}(s)$ in \eqref{eqn: seq root}, we obtain $s^4=0$. For repeated roots with multiplicity 4, we apply \cite[Eq. 22]{nie_bandwidth_2017-1} to derive the  bounds 
					\begin{subequations}\label{eqn: bf chu antenna}
						\begin{alignat}{3}
							&\int_0^{\infty} \frac{1}{2 \pi^2 f^2} \log\left( \frac{1}{1 - \cT( f )}\right) \mathrm{d}f \leq \frac{2a}{c}, \label{eqn: bf chu antenna1} \\ 
							& \int_0^{\infty} \frac{1}{8\pi^4 f^4} \log\left( \frac{1}{1 - \cT(f )}\right) \mathrm{d}f \leq \frac{4a^3}{3c^3}.\label{eqn: bf chu antenna2}
						\end{alignat}
					\end{subequations}
					The bounds are in the form specified in \eqref{problem1_a} with $N_{\mathsf{BF}}=2$.
					
					\begin{figure}[htbp]
						\centering
						\begin{subfigure}[htbp]{0.5\textwidth}   
							\centering 
							\includegraphics[width=\textwidth]{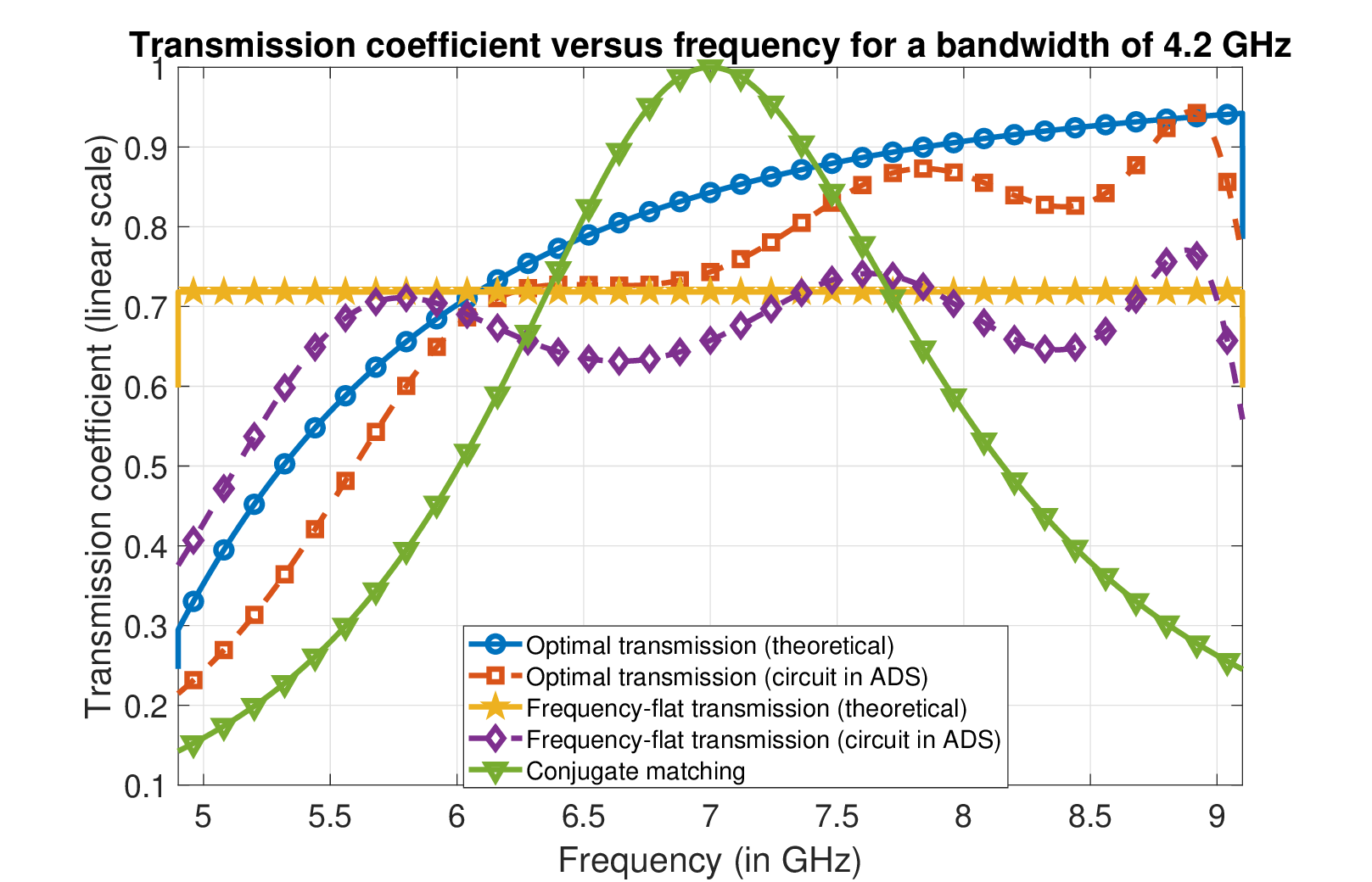}
							\caption
							{{Bandwidth of 4.2 GHz. }}      
							\label{fig: transmission coeff chu antenna 4200}
						\end{subfigure}
						\hfill
						\begin{subfigure}[htbp]{0.5\textwidth}   
							\centering 
							\includegraphics[width=\textwidth]{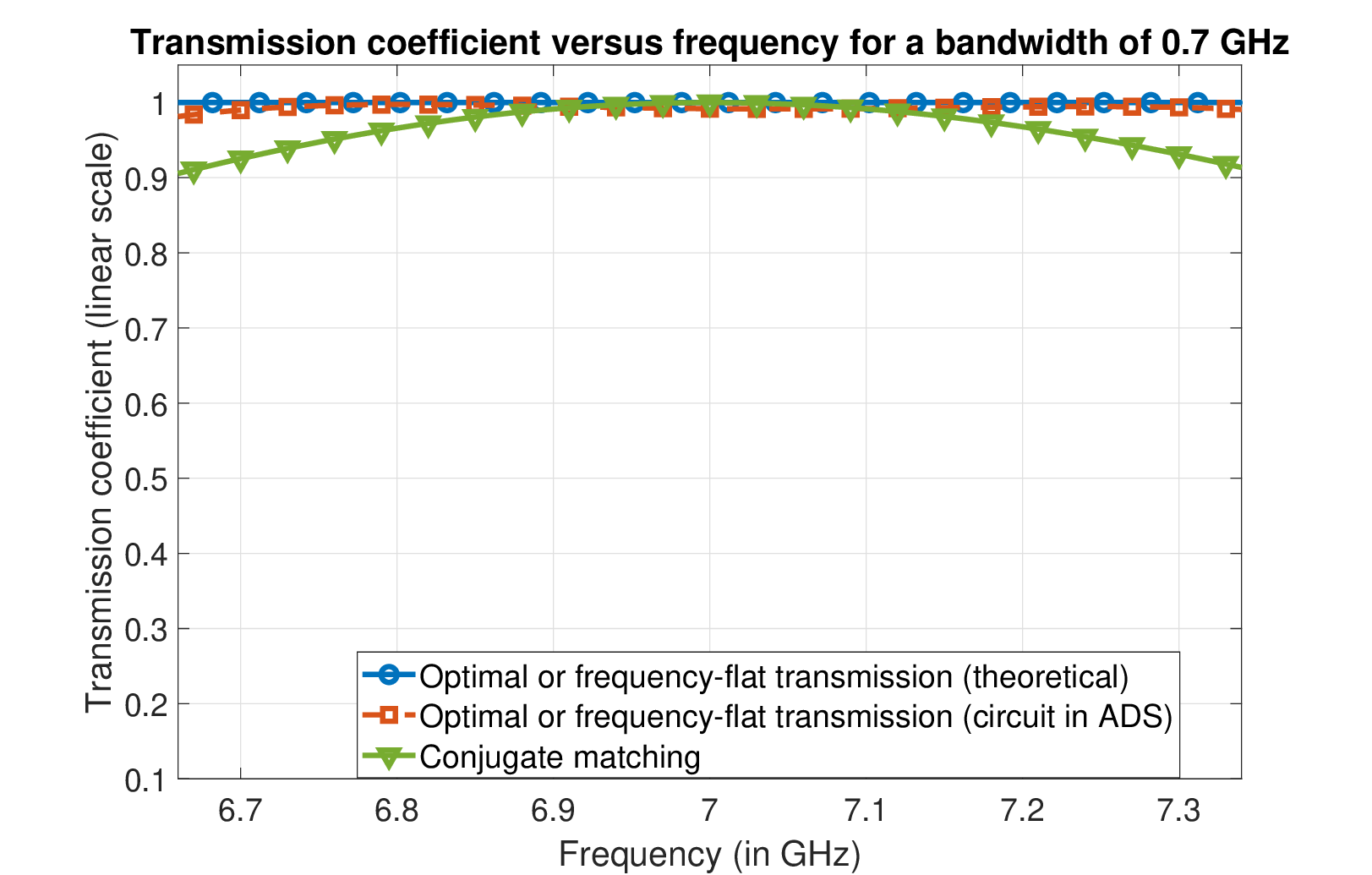}
							\caption
							{{ Bandwidth of 0.7 GHz. }} 
							\label{fig: transmission coeff chu antenna 700} 
						\end{subfigure}
						\caption{ For (a), the optimal transmission curve is higher than the frequency-flat transmission curve for frequencies greater than 6.1 GHz. For (b), the optimal transmission coincides with the frequency-flat transmission curves.							
							Conjugate matching  response is same for both bandwidths whereas the proposed optimal transmission and frequency-flat transmission responses change with the bandwidth.}
						\label{fig: transmission coeff chu antenna}
					\end{figure}
					
				\noindent	\textbf{Parameter setup}:				
						Let $f_{\mathsf{c}}= 7$ GHz. The corresponding wavelength $\lambda_{\mathsf{c}}=4.29 $ cm. Let $a=\frac{\lambda_{\mathsf{c}}}{10} = 4.29$ mm, bandwidth $B \in  \{0.1 f_{\mathsf{c}}, 0.6 f_{\mathsf{c}}\} = \{0.7, 4.2\}$ GHz,  $\mathsf{E}_{\mathsf{s}}= \frac{0.25}{B}\left[\frac{\text{W}}{\text{Hz}}\right] $, and $R=50 \Omega$. Let the distance between receiver and transmitter be $d_{\mathsf
						tx-rx}= 500$m and the  antenna gain $G=1.5$. 
					The wireless channel is  known at the transmitter and defined as 
					$\sfS_{\mathsf{RT}}(f)= \frac{1-\sfS_{\mathsf{T}}(f) }{\sfZ_0 + Z_{\mathsf{R}}(f) } \frac{c}{2 \pi f d_{\mathsf
							tx-rx}} G \cR(Z_{\mathsf{T}}( f))$\cite{wallace_mutual_2004}.
					Using  $k_B= 1.38 \times 10^{-23}$ J/K and $T=290 K$, we get
					$\mathsf{N}_0= 4 \times  10^{-21}\left[\frac{\text{W}}{\text{Hz}}\right]$.
					
				\noindent	\textbf{Optimal transmission coefficient approximation with an  LC ladder}:					
					We compute the optimal transmission coefficient by solving the achievable rate optimization problem using  two Bode-Fano bounds in \eqref{eqn: bf chu antenna}. 
					To approximate this transmission coefficient, it suffices to use a simple 4th order LC ladder shown in Fig.~\ref{fig: SISO schematic order 4 MN}. 
					We use Keysight ADS which is a circuit simulation software for characterizing and optimizing RF systems.
					In ADS, we define 8 design variables, $L_i$ and $C_i$  for $ \{i\} _{1}^{4}$, and set the optimization goal based on the optimal transmission coefficient. The output is the optimized values of $L_i$ and $C_i$. The transmission coefficient corresponding to the optimized circuit is used as a comparison benchmark.


					

				\noindent	\textbf{Frequency-flat transmission coefficient approximation with an LC ladder}:					
					For comparison with the box-car matching approach in \cite{taluja_diversity_2013}, 
					we assume frequency-flat transmission coefficient in a band spanning $f_{\mathsf{min}}$ to $f_{\mathsf{max}}$, i.e.,  $\cT(f)=\cT_{\mathsf{ff}}$ for $f\in [f_{\mathsf{min}}, f_{\mathsf{max}}]$. 
					The value of $\cT_{\mathsf{ff}}$ 
				should satisfy	both \eqref{eqn: bf chu antenna1} and \eqref{eqn: bf chu antenna2}.
					We define $	r_1= \exp\bigg(\frac{- 2a/c }{\int_{f_{\mathsf{min}}}^{f_{\mathsf{max}}}{1}/({2 \pi^2 f^2})\mathrm{d}f} \bigg)$ and  $	r_2= \exp\left(\frac{- {4a^3}/({3c^3}) }{\int_{f_{\mathsf{min}}}^{f_{\mathsf{max}}}{1}/({8\pi^4 f^4})\mathrm{d}f} \right)$.
					The value of $\cT_{\mathsf{ff}}$ satisfying both  constraints 
					is  $\cT_{\mathsf{ff}} = 1- \mbox{Max}\{r_1, r_2\}$.
					The frequency-flat transmission coefficient is approximated similarly using a 4th order LC ladder 
					in ADS.


					In Fig.~\ref{fig: transmission coeff chu antenna}, we plot the  transmission coefficient versus frequency for the theoretical case and the 4th order LC ladder circuit in Fig.~\ref{fig: SISO schematic order 4 MN} optimized in ADS.
					We see that the circuit implemented in ADS provides a good approximation of the desired transmission coefficient in the bandwidth of interest. This shows that with a simple matching network topology, it is possible to approximate the  transmission coefficient. 
					We also show the conjugate matching transmission coefficient benchmark in Fig.~\ref{fig: transmission coeff chu antenna} which remains the same irrespective of the bandwidth. The proposed optimal transmission coefficient curves are dependent on the bandwidth. In  Fig.~\ref{fig: transmission coeff chu antenna 4200}, although conjugate matching transmission has higher $\cT(f)$  than the optimal transmission  for frequencies  6.5 GHz to 7.4 GHz, this is because the optimal transmission is optimized for a larger band from 4.9 GHz to 9.1 GHz. In  Fig.~\ref{fig: transmission coeff chu antenna 700}, the optimal transmission  is optimized for 6.65 GHz to 7.35 GHz and is higher than  fixed conjugate matching. 					
					We use these transmission coefficients in Section~\ref{sec: numerical results} to compute the SNR, achievable rate, and its comparison with other  benchmarks.

					\subsection{Application of  matching network design methodology to an array of two Chu's antennas}\label{subsec: two dipoles}
					 We present numerical illustrations for the matching network design methodology applied to an array of two Chu's antennas.
					
				\noindent	\textbf{Chu's antenna array model}:					
					We assume an array of two parallel Chu's antennas, each enclosed in a spherical volume of radius $a$ and  separated by a distance $d$. 
					The self impedance for each antenna is $\sfZ_{\mathsf{T11}}(f) = \sfZ_{\mathsf{T22}}(f)= \frac{R}{\sfj 2 \pi  f \frac{a}{c}}+ \frac{R}{1+\frac{1}{\sfj 2 \pi  f \frac{a}{c}}}$.
					The mutual impedance between  two antennas is  \cite{akrout2022super}
					\begin{align}
						&\sfZ_{\mathsf{T12}}(f)=\sfZ_{\mathsf{T21}}(f) = -1.5 \sqrt{\cR(\sfZ_{\mathsf{T11}}(f)) \cR(\sfZ_{\mathsf{T22}}(f))}  \\ \nonumber &\times\left(\!\frac{1}{\sfj 2 \pi  f \frac{d}{c}}    - \frac{1}{(2 \pi f \frac{d}{c})^2} +\frac{j}{ ( 2 \pi f \frac{d}{c})^3} \! \right)\! e^{-\sfj 2 \pi  f \frac{d}{c}}.
					\end{align}
					The  array impedance matrix is defined as $\bm{\sfZ}_{\mathsf{T}}(s) = \begin{bmatrix}
						\sfZ_{\mathsf{T11}}(s) & \sfZ_{\mathsf{T12}}(s) \\
						\sfZ_{\mathsf{T21}}(s) & \sfZ_{\mathsf{T22}}(s)
					\end{bmatrix}.$  The scattering matrix is $	\bm{\sfS}_{\mathsf{T}}(s)= (\bm{\sfZ}_{\mathsf{T}}(s) + R \bI_2)^{-1} (\bm{\sfZ}_{\mathsf{T}}(s) - R \bI_2)$.
					
				\noindent	\textbf{Analog beamforming network model}:					
					For the scattering matrix of the analog beamforming network,
					we assume $	\sfS_{{\mathsf{F}},11}(f)=0$, $\bm{\sfS}_{{\mathsf{F}},22}(f)=\mathbf{0}_2 $, and $\bm{\sfs}_{{\mathsf{F}},12}(f) = \bm{\sfs}_{{\mathsf{F}},21}^T(f)$ represents the beamforming vector corresponding to an ideal frequency-flat phased array.
					We simulate  two beamforming modes similar to \cite{nie_bandwidth_2017-1}. The even mode corresponds to $\bm{\sfs}_{{\mathsf{F}},21}(f) = \frac{-\sfj}{\sqrt{2}}[1, 1]^T$ and the odd mode corresponds to 
					$\bm{\sfs}_{{\mathsf{F}},21}(f) = \frac{-\sfj}{\sqrt{2}}[1, -1]^T$. 		
					We also assume no insertion loss.
						\begin{figure}[htbp]
							\centering
							\begin{subfigure}[htbp]{0.5\textwidth}   
								\centering 
								\includegraphics[width=\textwidth]{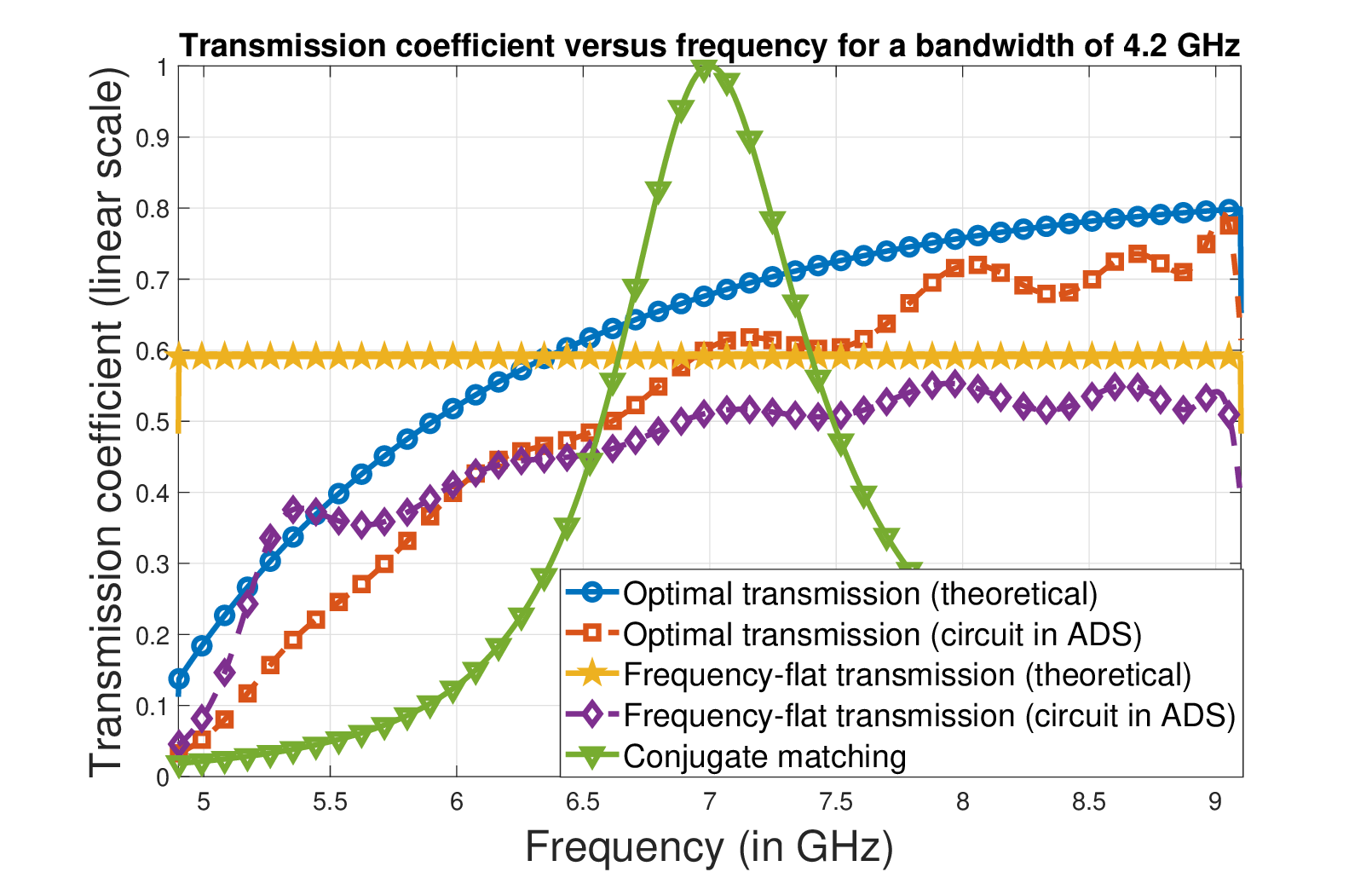}
								\caption
								{{Odd mode beamforming for $\theta=\frac{\pi}{2}$ }}      
								\label{fig: tx miso odd}
							\end{subfigure}
							\hfill
							\begin{subfigure}[htbp]{0.5\textwidth}   
								\centering 
								\includegraphics[width=\textwidth]{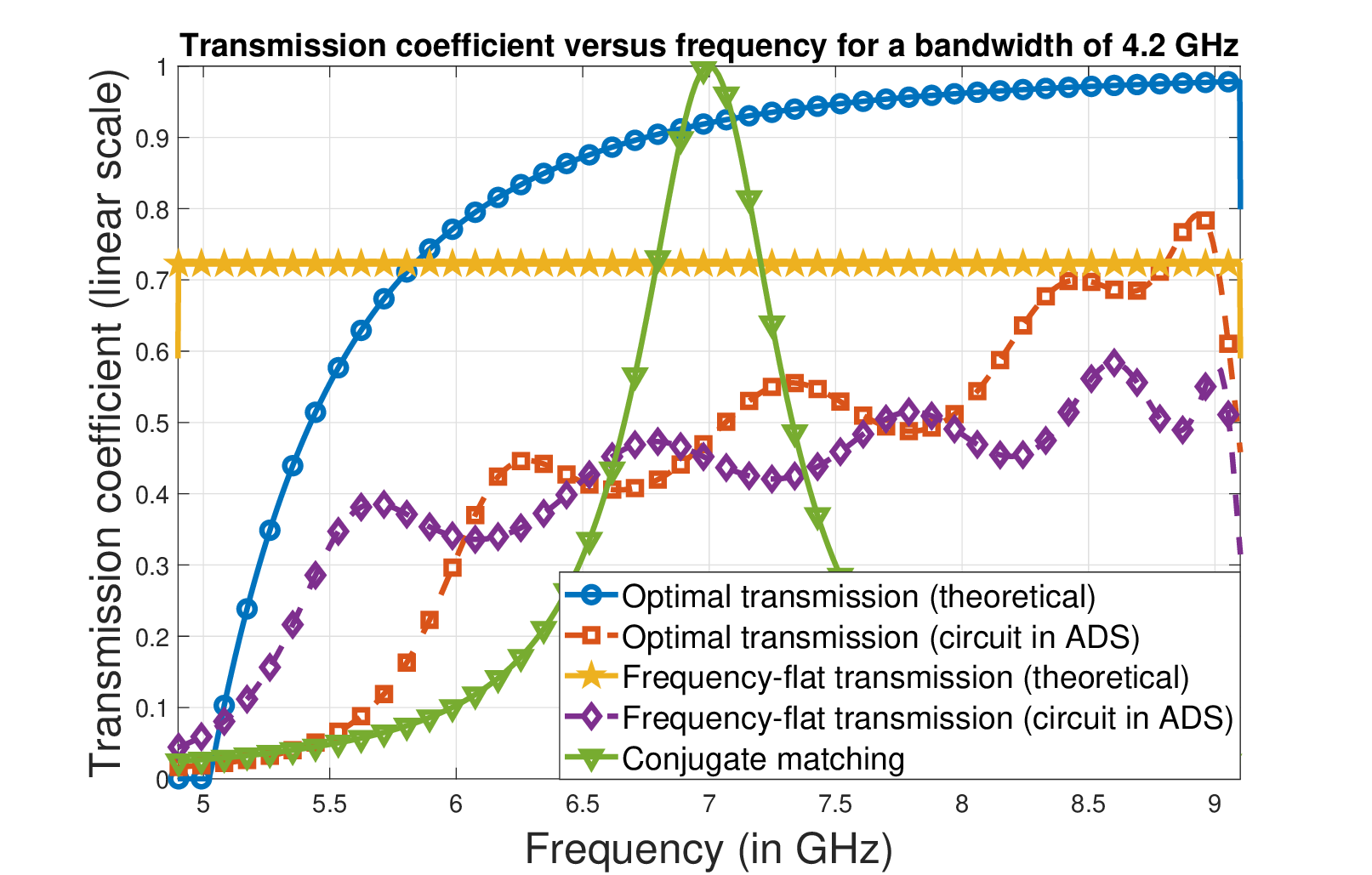}
								\caption
								{{ Even mode beamforming for $\theta=0$ }} 
								\label{fig: tx miso even} 
							\end{subfigure}
							\caption{For an array of two Chu's antennas with even and odd mode analog beamforming, the optimal transmission curves are higher than the frequency-flat transmission curves for a major portion of the bandwidth.  As matching network order is fixed for both beamforming modes, the gap between theoretical and circuit  response is different.
							}
							\label{fig:tx miso even odd }
						\end{figure}

					\noindent	\textbf{Parameter setup}:						
						For $f_{\mathsf{c}}= 7$ GHz, let $d=\frac{\lambda_{\mathsf{c}}}{2} $  and $a=\frac{\lambda_{\mathsf{c}}}{15} $,
						 bandwidth $B= 0.6 f_{\mathsf{c}} = 4.2$ GHz, and  $\mathsf{E}_{\mathsf{s}}= \frac{0.25}{B}\left[\frac{\text{W}}{\text{Hz}}\right] $.
						For receiver at angle $\theta$ from  broadside,
						the wireless channel is \cite{wallace_mutual_2004}
						\begin{equation}\label{eqn sRT}
						\bm{\sfs}_{\mathsf{RT}}(f) =  \frac{c G \cR(\sfZ_{\mathsf{T}}( f))   (\bI_2-	\bm{\sfS}_{\mathsf{T}}(f)  )[1, \exp(\sfj 2 \pi  f \frac{d}{c} \sin(\theta))]^T}{2 \pi f d_{\mathsf
								tx-rx} (\sfZ_0 + \sfZ_{\mathsf{R}}(f))} .
						\end{equation}
					In \eqref{eqn sRT}, the mutual coupling effect between antennas is captured through the dependence on $	\bm{\sfS}_{\mathsf{T}}(f)$.

						\noindent\textbf{Bode-Fano bounds for Chu's antenna array}:						
						For both even and odd beamforming modes, we compute  $\hat{\sfS}_{\mathsf{eq}}(s)$ 
						and substitute it in \eqref{eqn: seq root} to solve for $s$. We obtain two unique roots with $\cR(s_0)>0$ which correspond to two Bode-Fano inequalities computed  using  Table~\ref{tab: eval bf functions} in Appendix~\ref{app: fbfi bfi}.


					\noindent	\textbf{Optimal transmission coefficient approximation with an LC ladder}:
						We compute the optimal transmission coefficient by solving the achievable rate optimization problem using  two Bode-Fano constraints. 
						For approximating  $\cT^{\star}( f ) $, it suffices to choose a 7th order LC ladder. In ADS, we define 14 design variables: $L_i$ and $C_i$ for $ \{i\} _{1}^{7}$ and set the optimization goal based on  $\cT^{\star}( f ) $. Note that the choice of model order can be changed depending on other design requirements. 
						
						\noindent \textbf{Frequency-flat transmission coefficient approximation with an LC ladder}:
							For comparison with the box-car matching approach in \cite{taluja_diversity_2013}, 
						let $\cT_{\mathsf{ff}}$  satisfy both Bode-Fano constraints. Using  Table~\ref{tab: eval bf functions} in Appendix~\ref{app: fbfi bfi}, for $ \{i\} _{1}^{2} $, we define 
						\begin{equation}
							r_i= \exp\left(\!\!\frac{\log\left(\left|\hat{\sfS}_{\mathsf{eq}}(s_{i})\frac{\prod_{\ell=1}^{N_{\mathsf{z}}} (s_{i} + z_{\mathsf{eq}, \ell})}{\prod_{\ell=1}^{N_{\mathsf{z}}} (s_{i} - z_{\mathsf{eq}, \ell})}\right|\right)}{\int_{f_{\mathsf{min}}}^{f_{\mathsf{max}}}\cR\{ (s_{i}- \sfj 2 \pi  f)^{-1} +(s_{i}+\sfj 2 \pi  f)^{-1} \} \mathrm{d}f}\!\!\right)\!\!.
						\end{equation}
						The value of $\cT_{\mathsf{ff}}$ satisfying both  constraints is  $\cT_{\mathsf{ff}} = 1- \mbox{Max}\{r_1, r_2\}$. The frequency-flat transmission coefficient can  be approximated similarly using a 7th order LC ladder in ADS.
						
						In Fig.~\ref{fig:tx miso even odd }, we plot the  transmission coefficient versus frequency for the theoretical case and the 7th order LC ladder circuit  optimized in ADS. For $\theta=\frac{\pi}{2}$, i.e., endfire incidence, we use the odd mode beamforming as shown in Fig.~\ref{fig:tx miso even odd }(a), whereas even mode beamforming is used for the broadside incidence as shown in Fig.~\ref{fig:tx miso even odd }(b).
						We observe that for a major portion of the 4.2 GHz bandwidth for both beamforming modes, optimal transmission response is higher than   frequency-flat transmission  response.
						This leads to a higher SNR and achievable rate.
						

						\section{SNR and achievable rate  simulations}\label{sec: numerical results}

						In this section, we present the results for  SNR and achievable rate corresponding to the following six cases.
							\begin{figure}[htbp]
							\centering
							\begin{subfigure}[htbp]{0.5\textwidth}   
								\centering 
								\includegraphics[width=\textwidth]{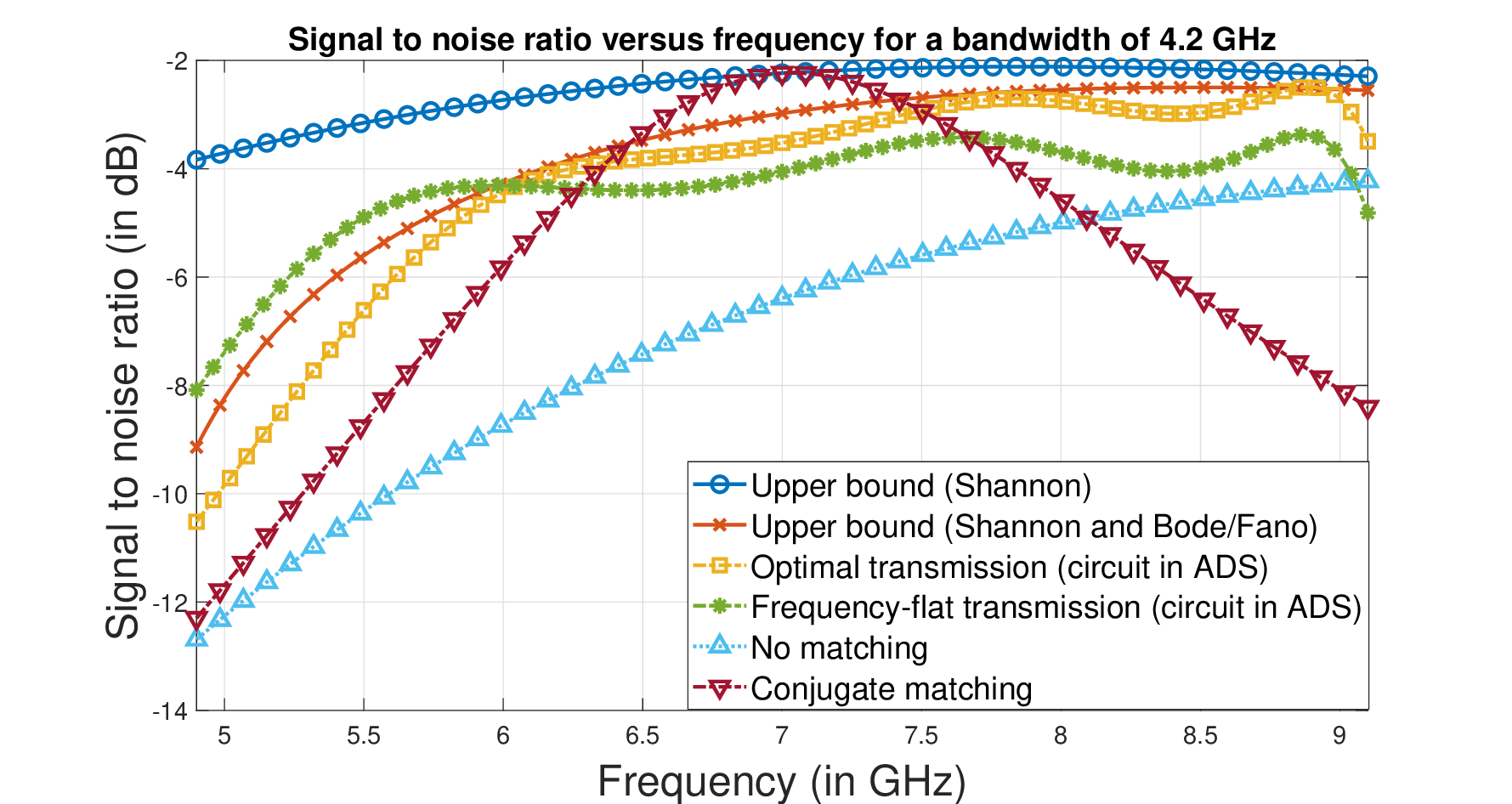}
								\caption
								{{Bandwidth of 4.2 GHz.  }}      
								\label{fig: SNR single antenna 4200}
							\end{subfigure}
							\hfill
							\begin{subfigure}[htbp]{0.5\textwidth}   
								\centering 
								\includegraphics[width=\textwidth]{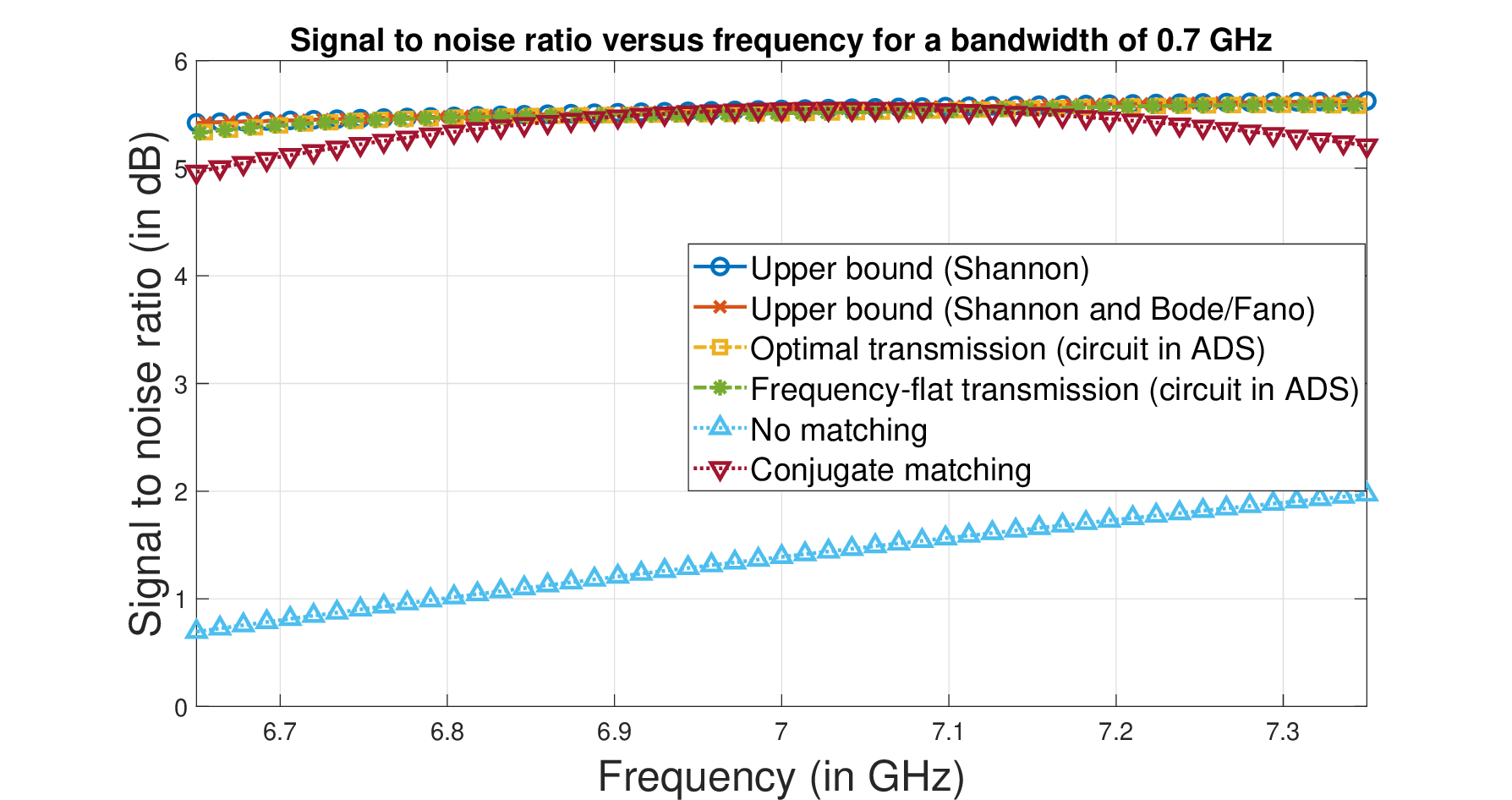}
								\caption
								{{ Bandwidth of 0.7 GHz.  }} 
								\label{fig: SNR single antenna 700} 
							\end{subfigure}
							\caption{For a single Chu's antenna -  In (a), the SNR curve for the optimal transmission ADS circuit is higher than the frequency-flat transmission curve for frequencies above 6.1 GHz. In (b), the SNR curves for the optimal and frequency-flat transmission coincide. }
							\label{fig: SNR single antenna}
						\end{figure}
					
						\noindent	\textbf{Upper bound (Shannon)}:  This bound corresponds to the case when the Bode-Fano constraints are ignored, i.e., $\cT(f)=1 \forall f$.  
							The ideal SNR denoted as  $\mathsf{SNR}_{\text{ideal}}(f)$ is defined in 
							\eqref{eqn: SNR ideal}.
							
							\noindent	\textbf{Proposed upper bound (Shannon and  Bode/Fano)}:  This upper bound is obtained as a  solution to problem \eqref{problem1} which maximizes the achievable rate over all physically realizable matching networks at  transmitter.
							
							\noindent \textbf{Conjugate matching at center frequency}:  The matching network is designed such that the load impedance gets transformed to the complex conjugate of the source impedance.
							
							\noindent \textbf{Proposed optimal transmission (circuit in ADS)}:  This case corresponds to the optimized matching network design obtained through the three step procedure  in Section~\ref{subsec: general methodology}.
							
							\noindent \textbf{Benchmark of frequency-flat transmission (circuit in ADS)}: 
							The matching network  is designed to approximate the  frequency-flat coefficient $\cT_{\mathsf{ff}}$  satisfying  Bode-Fano constraints.
							
								\noindent \textbf{Benchmark of no matching}: Finally, we also compare with the case when matching network is absent, i.e., the source is directly connected to the antenna array.
							The resulting SNR is defined as $\mathsf{SNR}_{\text{No-match}}(f)= {|\bm{\sfs}_{\mathsf{RT}}^T(f)  (\bI-\bm{\sfS}_{{\mathsf{F}},22}(f)	\bm{\sfS}_{\mathsf{T}}(f) )^{-1}	\bm{\sfs}_{{\mathsf{F}},21}(f)|^2}  \frac{ \mathsf{E}_{\mathsf{s}}}{\mathsf{N}_0}$.

						\subsection{Simulation results for a single Chu's antenna}\label{subsec: sim result single antenna}

						In Fig.~\ref{fig: SNR single antenna}(a),  we use the  parameter setup and the transmission coefficient based on circuit design in Section~\ref{subsec: single dipole} for
						computing the SNR versus frequency for the six cases. We observe that the SNR solely based on Shannon's upper bound  is higher than that of the bound which incorporates Bode-Fano wideband matching theory.
						The Shannon upper bound technique overestimates the SNR. This upper bound cannot be attained by any practical matching network.  The bound proposed by combining Shannon's theory and Bode-Fano  theory is more realistic as it incorporates the gain-bandwidth tradeoff in matching networks. We show that this bound  can be approximated   using a practical matching network topology optimized using ADS as discussed in Section~\ref{subsec: single dipole}.
						\begin{figure}[htbp]
							\centering
							\includegraphics[width=0.5\textwidth]{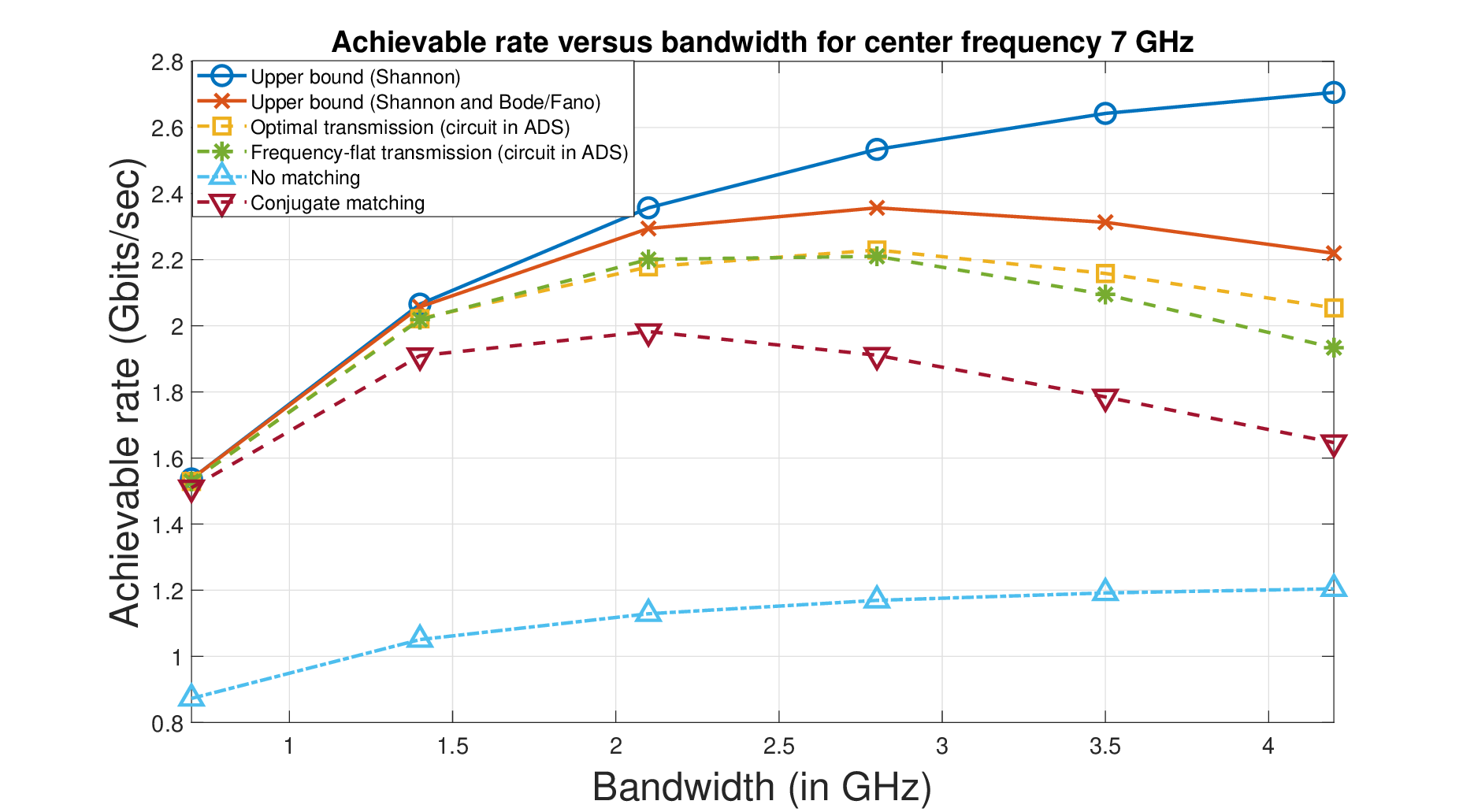}
							\caption{  Achievable rate versus bandwidth for a single Chu's antenna with center frequency 7 GHz.
								The  circuit based on optimal transmission significantly outperforms the frequency-flat transmission circuit for larger bandwidths. Also, the optimal bandwidth for getting the highest achievable rate is 2.8 GHz beyond which rate decreases.}
							\label{fig: rate vs bw single antenna}
						\end{figure}
						For a major portion of the 4.2 GHz bandwidth, the SNR for the ADS circuit corresponding to the optimal transmission is greater than that of the frequency-flat transmission circuit. This leads to a higher achievable rate as well for the optimal transmission based circuit.
					The SNR corresponding to the conjugate matching network  is higher than the optimal transmission for frequencies 6.5 GHz to 7.4 GHz but drastically decreases outside this range.
						As rate depends on the SNR for the whole band from 4.9 GHz to 9.1 GHz, the  achievable rate for conjugate matching is less than the rate for optimal transmission and frequency-flat transmission.
					Even for a bandwidth of 0.7 GHz as shown in Fig.~\ref{fig: SNR single antenna}(b),  conjugate matching is still worse compared to the proposed approach because the optimal transmission  solution in \eqref{eqn: thm1} depends on  bandwidth unlike conjugate matching as shown  in Fig.~\ref{fig: transmission coeff chu antenna}.

					\begin{figure}[htbp]
						\centering
						\begin{subfigure}[htbp]{0.5\textwidth}   
							\centering 
							\includegraphics[width=\textwidth]{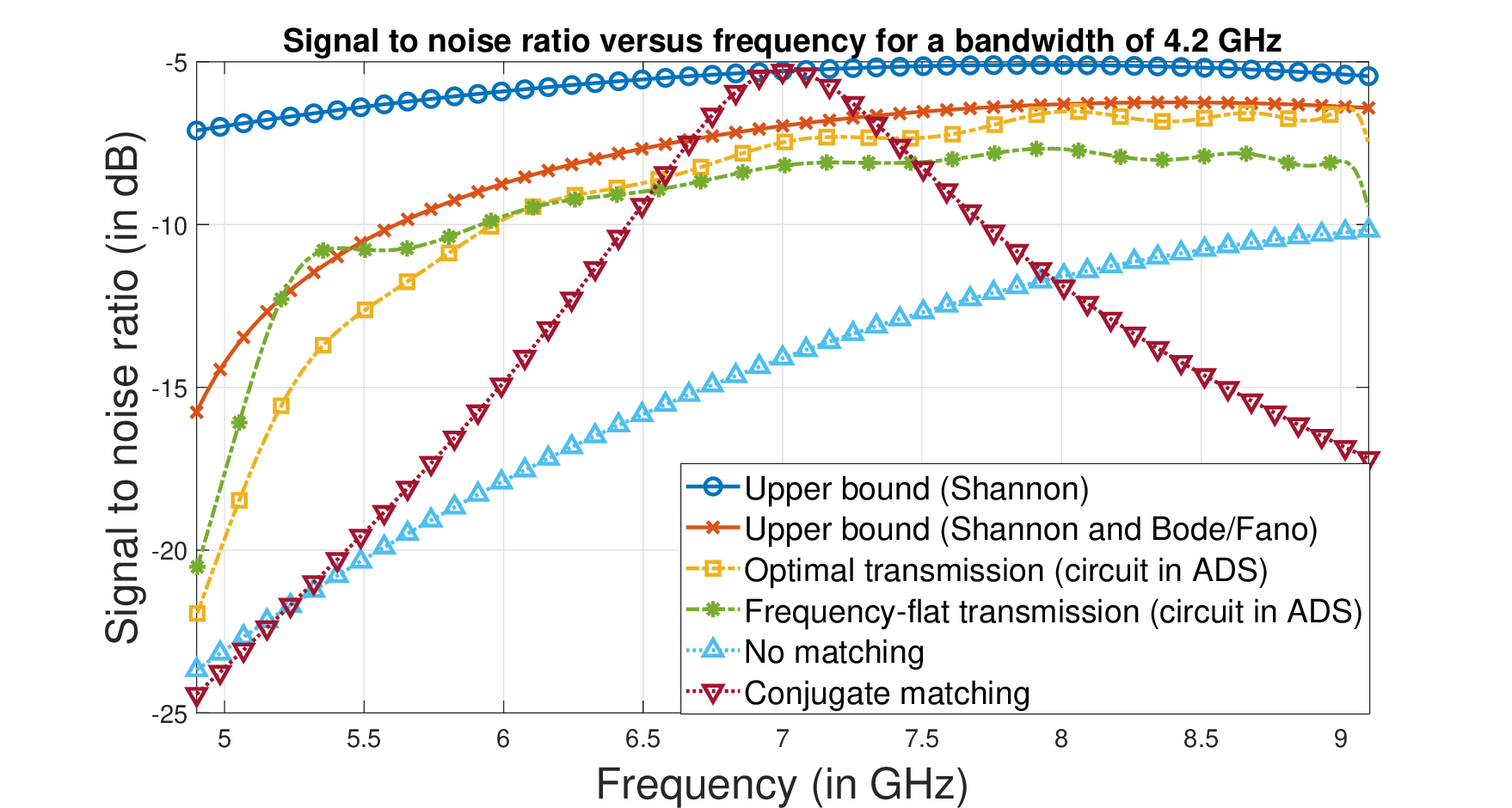}
							\caption
							{{Odd mode beamforming for $\theta=\frac{\pi}{2}$ }}      
							\label{fig: snr miso odd}
						\end{subfigure}
						\hfill
						\begin{subfigure}[htbp]{0.5\textwidth}   
							\centering 
							\includegraphics[width=\textwidth]{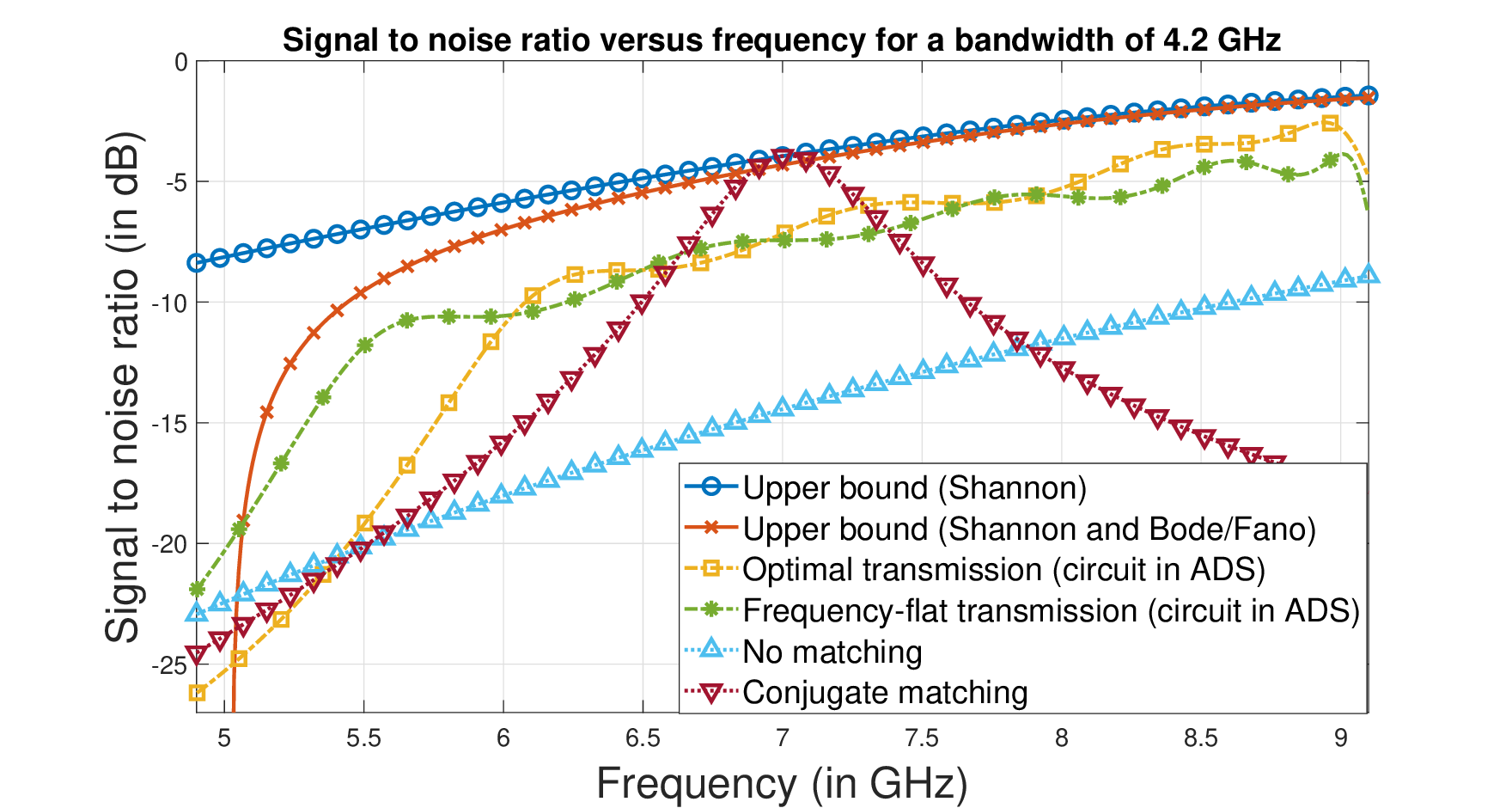}
							\caption
							{{ Even mode beamforming for $\theta=0$ }} 
							\label{fig: snr miso even} 
						\end{subfigure}
						\caption{For an array of two Chu's antennas with frequency-flat even and odd mode analog beamforming. 
							SNR degradation for   $\theta=\frac{\pi}{2}$ compared to $\theta=0$  because of beam squint effect.
						}
						\label{fig:snr miso even odd }
					\end{figure}


						
						In Fig.~\ref{fig: rate vs bw single antenna},  
						we plot the achievable rate as a function of bandwidth.
						The achievable rate plot based on the Shannon upper
						bound continuously increases with the bandwidth. A realistic
						trend is observed for the upper bound obtained after incorporating
						Bode-Fano theory, i.e., we observe that the rate decreases beyond a certain bandwidth.
						The results highlight how the matching network limits the bandwidth and achievable rate of the system.
	We show  the existence of an optimal bandwidth which gives the highest possible achievable rate because of the gain-bandwidth tradeoff of matching networks.
						From the proposed upper bound plot and the corresponding ADS circuit simulation, the optimal bandwidth for highest achievable rate is 2.8 GHz. For conjugate matching, this optimal bandwidth is 2.1 GHz. 
						It is lower than the optimal bandwidth of the proposed ADS circuit because conjugate matching response is invariant of the bandwidth.
						The proposed matching approach accounts for the bandwidth dependence and enables use of higher bandwidth for maximizing  rate.  

						\begin{figure}[htbp]
							\centering
							\begin{subfigure}[htbp]{0.5\textwidth}   
								\centering 
								\includegraphics[width=\textwidth]{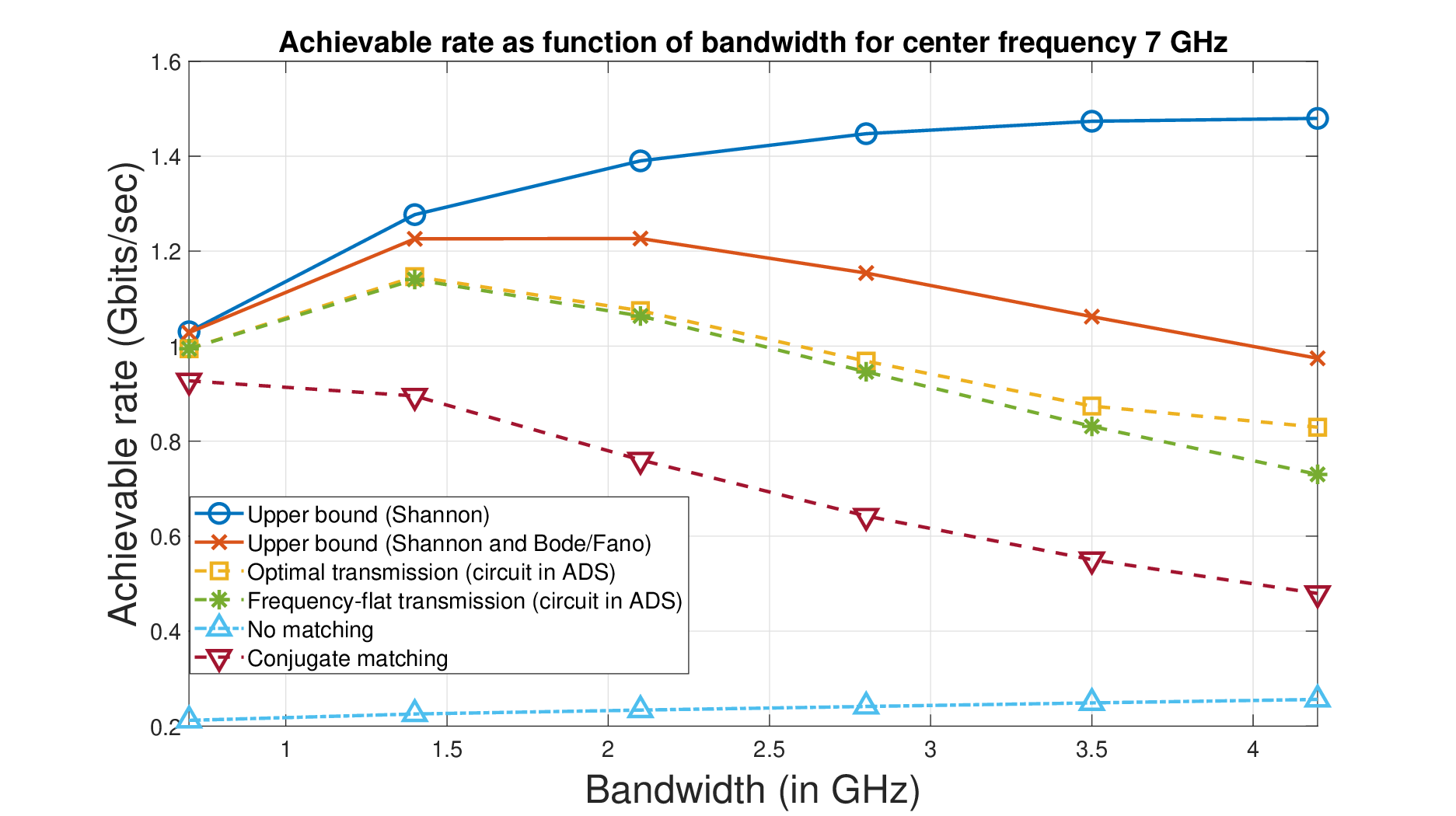}
								\caption
								{{Odd mode beamforming for $\theta=\frac{\pi}{2}$ }}      
								\label{fig:rate odd mode}
							\end{subfigure}
							\hfill
							\begin{subfigure}[htbp]{0.5\textwidth}   
								\centering 
								\includegraphics[width=\textwidth]{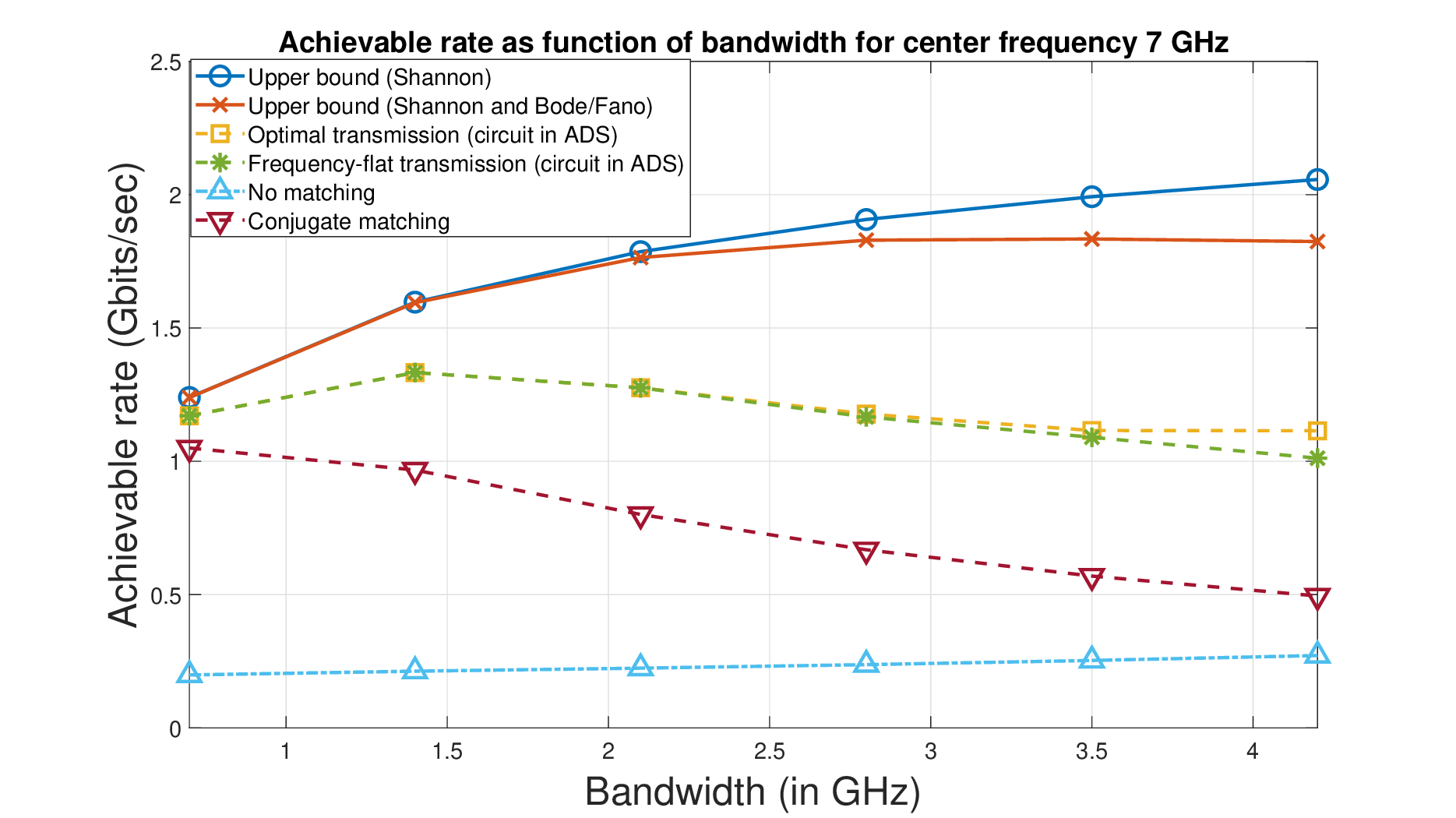}
								\caption
								{{Even mode beamforming for $\theta=0$}} 
								\label{fig:rate even mode} 
							\end{subfigure}
							\caption{Achievable rate versus bandwidth for an array of two Chu's antennas with even and odd mode analog beamforming. From the circuit based plots for both  beamforming modes, the achievable rate peak occurs at a bandwidth of 1.4 GHz. At larger bandwidths, odd mode achievable rates lower than even mode because of beam squint.}
							\label{fig:rate miso even odd}
						\end{figure}

						\subsection{Simulation results for an array of two Chu's antennas}

						In Fig.~\ref{fig:snr miso even odd },  we use the  parameter setup and  the transmission coefficient based on circuit design in Section~\ref{subsec: two dipoles} for
						computing the SNR versus frequency for the six cases.
						Most of the SNR comparison and trends are similar to that discussed in Section~\ref{subsec: sim result single antenna}.  
						At broadside incidence, i.e.,  $\theta=0$, there is no relative phase difference between two antennas. Hence, even mode beamforming works well. For endfire incidence, i.e., $\theta=\frac{\pi}{2}$, the phase-difference between two antennas varies as a function of frequency as shown in \eqref{eqn sRT}. 
						Using a frequency-flat beamforming for endfire at higher bandwidths results in a phase mismatch and subsequent SNR reduction. This effect is commonly known as beam squint.
						So, we observe SNR degradation for  endfire compared to broadside in Fig.~\ref{fig:snr miso even odd }.

						In Fig.~\ref{fig:rate miso even odd}, 
						for the odd mode, we observe that the upper bound obtained after incorporating Bode-Fano theory significantly deviates from the Shannon bound beyond 1.4 GHz bandwidth. This trend is also consistent with the rate obtained from the circuit simulations in ADS.
						For the even mode, the achievable rate increases faster with bandwidth compared to odd mode because there is no phase mismatch with even mode. Frequency-selective true time delay (TTD) beamforming can be used to mitigate the beam squint effect. The achievable rate variation with bandwidth for TTD  systems is a future direction.


						\section{Conclusion}\label{sec: conclusion}
						
						In this paper, we generalized the achievable rate analysis for a MISO system by incorporating constraints from Bode-Fano wideband matching theory.
						We proposed a general optimization framework which maximizes the achievable rate over all physically realizable linear and passive matching networks.   The proposed upper bound based on the combination of Shannon's theory and Bode-Fano  theory is more realistic because it captures the gain-bandwidth tradeoff of matching networks. 
						   We also proposed a simple three step procedure to design matching networks that approximate this bound. 
						We demonstrated this procedure for a single Chu's antenna and an array with two Chu's antennas. 
						From the derived theoretical bound and the ADS circuit simulations,  an optimal bandwidth behavior is observed in the achievable rate analysis as function of the bandwidth.
							In future work, we propose the application of this methodology  to other antenna types like dipoles or patch antenna. The main challenge is to numerically compute the upper bound because the number of Bode-Fano constraints increases for complicated antenna geometries.  In future work, we plan to extend this work to MIMO systems with  multiple RF chains  and study the bandwidth-multiplexing tradeoff\cite{nie_bandwidth_2017-1}. The main challenge in extension to MIMO is the joint optimization of multiport transmit and receive impedance-matching networks.

						\appendices
							\begin{table*}[htbp]
							\centering
							\caption{Evaluation of  $\xi_{\mathsf{BF}, i}(f)$ and $B_{\mathsf{BF},i}$  (Based on \cite[Table 1]{nie_bandwidth_2017}) }
							\label{tab: eval bf functions}
							{
								\begin{tabular}{|c|c|c|}
									\hline
									Location in WCP	&  $\xi_{\mathsf{BF}, i}( f)$ &   $B_{\mathsf{BF},i}$  \\ \hline
									$s_{i}=\sfj 2 \pi  f_{i} $ 	 & $\frac{1}{4 \pi^2}  \left[\frac{1}{(f_i - f)^2}+ \frac{1}{(f_i+ f)^2}\right]$ &  - $\left[\sum_{\ell=1}^{N_{\mathsf{p}}} (p_{\mathsf{eq}, \ell} - \sfj 2 \pi  f_i)^{-1}+ \sum_{m=1}^{N_{\mathsf{z}}} (z_{\mathsf{eq}, m} + \sfj 2 \pi  f_i)^{-1} \right]$ \\ \hline
									$\cR\{{s_{i}}\}>0$ 		  & $\cR\{  (s_i - \sfj 2 \pi  f_i)^{-1}  +  (s_i + \sfj 2 \pi  f_i)^{-1} \}$   & $-\log\left(\left|\hat{\sfS}_{\mathsf{eq}}(s_{i})\frac{\prod_{m=1}^{N_{\mathsf{z}}} (s_{i} + z_{\mathsf{eq}, m})}{\prod_{m=1}^{N_{\mathsf{z}}} (s_{i} - z_{\mathsf{eq}, m})}\right|\right)$  \\ \hline  
									$s_i= \infty$	&   1   & $\frac{-1}{2}  \left[ \sum_{\ell=1}^{N_{\mathsf{p}}} p_{\mathsf{eq}, \ell}  +  \sum_{m=1}^{N_{\mathsf{z}}} z_{\mathsf{eq}, m} \right]$    \\ \hline
							\end{tabular}}
						\end{table*}
						
						\section{Computing the  scattering parameter $\hat{\sfS}_{\mathsf{eq}}(s) $ in rational form \cite{nie_bandwidth_2017-1} }\label{app: Seq s}

					\noindent	 \textbf{Case 1}:  The impedance parameter of the load is analytically known in the rational form in the whole complex plane and denoted as $\hat{\sfZ}_{\mathsf{eq}}(s) $.  The corresponding scattering parameter of the load in the rational form is $\hat{\sfS}_{\mathsf{eq}}(s) = \frac{\hat{\sfZ}_{\mathsf{eq}}(s) - \sfZ_0}{\hat{\sfZ}_{\mathsf{eq}}(s) + \sfZ_0}$.
					
							 \noindent \textbf{Case 2}:   The measured value of the scattering parameter of the load   ${\sfS}_{\mathsf{eq}}(f) $ is available for the frequency $f$ in the range of interest $[f_1, f_2]$. A passive and rational approximation $\hat{\sfS}_{\mathsf{eq}}(s)$ is obtained such that $\hat{\sfS}_{\mathsf{eq}}(\sfj 2 \pi  f)$ is close to  ${\sfS}_{\mathsf{eq}}(f) $ for $f \in [f_1, f_2]$ within a specified error tolerance.
							This can be done numerically using the \textit{rationalfit} function in MATLAB\cite{gustavsen1999rational}.

						\section{Computing $\xi_{\mathsf{BF}, i}(f)$  and $B_{\mathsf{BF},i}$ for $\{i\} _{1}^{N_{\mathsf{BF}}} $ from $\hat{\sfS}_{\mathsf{eq}}(s) $ }\label{app: fbfi bfi}

						%
								%

						Using closed-form expression of $	\hat{\sfS}_{\mathsf{eq}}(s)$, we first solve  the following for $s$\cite{nie_bandwidth_2017}.
						\begin{equation}\label{eqn: seq root}
							\hat{\sfS}_{\mathsf{eq}}(-s)\hat{\sfS}_{\mathsf{eq}}(s)-1=0.
						\end{equation}
						Let $s_{i}$ be a distinct root of \eqref{eqn: seq root}. 
						The value of  $s_i$ in \eqref{eqn: seq root}  can be obtained analytically or numerically by using \textit{vpasolve} function in MATLAB.
						Let $\{z_{\mathsf{eq}, 1}, \dots, z_{\mathsf{eq}, m}, \dots,  z_{\mathsf{eq},N_{\mathsf{z}}} \}$ be the zeros and  $\{p_{\mathsf{eq}, 1}, \dots, p_{\mathsf{eq}, \ell}, \dots,  p_{\mathsf{eq},N_{\mathsf{p}}} \}$ be the poles of the rational equivalent load $	\hat{\sfS}_{\mathsf{eq}}(s)$. 
						For each $s_{i}$, there is a corresponding $\xi_{\mathsf{BF}, i}(f)$  and $B_{\mathsf{BF},i}$ depending on the location of $s_{i}$ in the whole complex plane (WCP) categorized in Table~\ref{tab: eval bf functions}.
						For the case of  multiplicity of $s_i$  more than one, i.e. there are repeated roots,   $\xi_{\mathsf{BF}, i}( f)$ and $B_{\mathsf{BF},i}$   can be computed using \cite[Eq 21- Eq 23]{nie_bandwidth_2017-1} for each repeated root.

						\section{Proof of Theorem 2}\label{app: proof thm2}
						
						The variables $\cT^{\star}(  f) $ and $\mu^{\star}_i |_{i=1}^{N_{\mathsf{BF}}+2} $   satisfy the  KKT conditions\cite{shyianov_achievable_2022,boyd2004convex}  applied to \eqref{problem1}.
						
						\textbf{Primal feasibility:}
						\begin{subequations}
							\begin{alignat}{3}
								&\int_0^{\infty} \!\!\!\xi_{\mathsf{BF},i}(  f) \log\left(\frac{1}{	1- \cT^{\star}(f)}\right) \mathrm{d}f - B_{\mathsf{BF},i} \leq 0	,   \{i\} _{1}^{N_{\mathsf{BF}}},  \label{eqn: pf1}\\
								&	\cT^{\star}( f)  - 1 \leq 0, \label{eqn: pf2}\\
								&  -\cT^{\star}(  f) \leq 0 \label{eqn: pf3}.
							\end{alignat}
						\end{subequations}
						
						\textbf{Dual feasibility:}  $	\mu_i^{\star} \geq 0.$

						\textbf{Complementary slackness:}
						\begin{subequations}
							\begin{alignat}{3}
								&\mu_i^{\star} \!\!\left(\int_0^{\infty} \!\!\!\xi_{\mathsf{BF},i}(  f) \log\left(\frac{1}{	1- \cT^{\star}(  f)}\right) \mathrm{d}f - B_{\mathsf{BF},i} \!\! \right)\!=\!0,  \{i\} _{1}^{N_{\mathsf{BF}}}\!\!\!\!\!,\label{eqn: cs1} \\
								& \mu_{N_{\mathsf{BF}}+1}^{\star} \cT^{\star}(  f)\! =0,\label{eqn: cs2}\\
								& \mu_{N_{\mathsf{BF}}+2}^{\star}(\cT^{\star}( f) -1)=0\label{eqn: cs3}.
							\end{alignat}
						\end{subequations}
						
						\textbf{Stationarity:}
						Let $\chi( f)$ be an arbitrary shaped function and $\epsilon$ represents the magnitude of variation \cite{shyianov_achievable_2022}. Using variational calculus and the stationarity condition, we set $\frac{\mathrm{d}}{\mathrm{d} \epsilon}	\left[\cL\left(\cT^{\star}( f)+ \epsilon \chi( f)\right)\right]\bigg|_{\epsilon=0}\!\!\!\!\!\!=0, $
						%
						simplified using chain rule as 
						\begin{align}\label{eqn: chain rule}
 &=\frac{\mathrm{d} \left(\cT^{\star}( f)+ \epsilon \chi( f)\right)}{\mathrm{d}\epsilon} \frac{\mathrm{d}\left[\cL\left(\cT^{\star}( f)+ \epsilon \chi( f)\right)\right]\bigg|_{\epsilon=0}}{\mathrm{d} \left(\cT^{\star}(f)+ \epsilon \chi(f)\right)}   \nonumber \\
							&=\frac{\mathrm{d}\left[\cL\left(\cT^{\star}( f)+ \epsilon \chi( f)\right)\right]}{\mathrm{d} \left(\cT^{\star}( f)+ \epsilon \chi( f)\right)} \chi( f)\bigg|_{\epsilon=0}\nonumber  \\
							&\stackrel{(a)}{=} \int_0^{\infty}\frac{\chi( f)}{\ln 2}\bigg[\left(\frac{-\mathsf{SNR}_{\text{ideal}}( f)}{1+ \mathsf{SNR}_{\text{ideal}}(  f)  \cT^{\star}(  f)}\right)  + \\ \nonumber &\sum_{i=1}^{N_{\mathsf{BF}}} \mu_i^{\star} \frac{\ln 2 \xi_{\mathsf{BF},i}( f)}{(1- \cT^{\star}( f))}\bigg]\mathrm{d}f  - \mu_{N_{\mathsf{BF}}+1}^{\star}\chi( f) + \mu_{N_{\mathsf{BF}}+2}^{\star}\chi( f),
						\end{align}
						where $(a)$ follows from \eqref{eqn: lagrangian}. Let $\mu_{N_{\mathsf{BF}}+1}^{\star}=0$ and $\mu_{N_{\mathsf{BF}}+2}^{\star}=0$ to satisfy \eqref{eqn: cs2} and \eqref{eqn: cs3} respectively.
						Equating \eqref{eqn: chain rule} to 0 is equivalent to setting the integrand inside the integral to 0. As $\chi( f)$ is an arbitrary function, the non-trivial condition  is 
						\begin{equation}\label{eqn: integrand is 0}
							\left[\left(\frac{-\mathsf{SNR}_{\text{ideal}}( f)}{1+ \mathsf{SNR}_{\text{ideal}}(f)  \cT^{\star}( f)}\right)  + \sum_{i=1}^{N_{\mathsf{BF}}} \mu_i^{\star} \frac{\ln 2 \xi_{\mathsf{BF},i}( f)}{(1- \cT^{\star}( f))}\right]=0.
						\end{equation}
						Simplifying \eqref{eqn: integrand is 0}, we obtain 	
						\begin{equation}\label{eqn stationarity}
							\cT^{\star}( f)\left(1+ \ln 2\sum_{i=1}^{N_{\mathsf{BF}}} \mu^{\star}_i \xi_{\mathsf{BF},i}(  f) \right) = 1- \frac{\ln 2 \sum_{i=1}^{N_{\mathsf{BF}}} \mu^{\star}_i \xi_{\mathsf{BF},i}( f)}{\mathsf{SNR}_{\text{ideal}}(  f) }.
						\end{equation}
						For \eqref{eqn: pf1} to be satisfied, $\cT^{\star}(  f)  - 1 \leq 0$ which requires atleast one $\mu_i^{\star}$
						to be strictly positive based on the expression of $\cT^{\star}(  f) $ from \eqref{eqn stationarity}. Combining \eqref{eqn: pf3} and \eqref{eqn stationarity}, we get \eqref{eqn: thm1}.

						\bibliographystyle{IEEEtran}
						\bibliography{references_Nitish.bib}

					\end{document}